\documentclass[11pt]{article}

\usepackage[margin=1in]{geometry}
\usepackage{amsmath,amssymb,amsthm,mathtools,bm}
\usepackage{booktabs,array,longtable}
\usepackage{enumitem}
\usepackage{xcolor}
\usepackage[numbers,sort&compress]{natbib}
\usepackage{hyperref}
\usepackage[nameinlink,capitalise]{cleveref}
\usepackage{aliascnt}
\usepackage{microtype}

\hypersetup{
  pdftitle={Classification and Exact Local Masking in Finite-Field Clifford Dual-Unitary Circuits},
  pdfauthor={Basanta R. Pahari},
  pdfsubject={Finite-field Clifford dual-unitary circuits, perfect tensors, local masking, and operator transport},
  pdfkeywords={Clifford circuits, dual-unitary gates, perfect tensors, quantum information masking, finite fields, operator transport},
  colorlinks=true,
  linkcolor=blue!55!black,
  citecolor=blue!55!black,
  urlcolor=blue!55!black
}

\newtheorem{theorem}{Theorem}[section]
\newaliascnt{proposition}{theorem}
\newtheorem{proposition}[proposition]{Proposition}
\aliascntresetthe{proposition}
\newaliascnt{lemma}{theorem}
\newtheorem{lemma}[lemma]{Lemma}
\aliascntresetthe{lemma}
\newaliascnt{corollary}{theorem}
\newtheorem{corollary}[corollary]{Corollary}
\aliascntresetthe{corollary}
\newaliascnt{definition}{theorem}
\newtheorem{definition}[definition]{Definition}
\aliascntresetthe{definition}
\newaliascnt{remark}{theorem}
\newtheorem{remark}[remark]{Remark}
\aliascntresetthe{remark}
\newaliascnt{conjecture}{theorem}
\newtheorem{conjecture}[conjecture]{Conjecture}
\aliascntresetthe{conjecture}

\newcommand{\F}{\mathbb F}

\newcommand{\Sp}{\operatorname{Sp}}
\newcommand{\SL}{\operatorname{SL}}
\newcommand{\GL}{\operatorname{GL}}
\newcommand{\rank}{\operatorname{rank}}
\newcommand{\tr}{\operatorname{tr}}
\newcommand{\wt}{\operatorname{wt}}

\newcommand{\id}{\mathrm{id}}
\newcommand{\cU}{\mathcal U}

\newcommand{\cN}{\mathcal N}

\newcommand{\cB}{\mathcal B}
\newcommand{\cE}{\mathcal E}
\newcommand{\cD}{\mathcal D}

\newcommand{\cL}{\mathcal L}
\newcommand{\norm}[1]{\left\lVert #1\right\rVert}
\newcommand{\Frob}[1]{\left\lVert #1\right\rVert_{\mathrm F}}

\title{Classification and Exact Local Masking in
Finite-Field Clifford Dual-Unitary Circuits}

\author{
Basanta R. Pahari\\[1mm]
\small Department of Mathematics, Hawai‘i Community College\\
\small University of Hawai‘i System, Hilo, Hawai‘i, USA\\
\small \texttt{basanta@hawaii.edu}
}

\date{June 2026}

\begin{document}
\maketitle

\begin{abstract}
We classify two-qudit Clifford dual-unitary gates over the finite field
$\mathbb{F}_q$, where the local dimension $q$ is a prime power, and apply
the classification to exact local masking and operator transport in
homogeneous brickwork circuits. Under ordered one-qudit Clifford
equivalence, the dual-unitary locus contains $q-2$ perfect-tensor
cores, one rank-one core, and one SWAP core. Homogeneous repetition
separates these cores into five distinct transport phases.

The one-site Weyl edge channels determine exact local-masking
distances. Writing $d_r(t)$ for the masking distance against output
observers controlling at most $r$ sites, perfect-tensor circuits attain
\[
d_1(t)=4t,
\qquad
d_2(t)=4t-2,
\]
whereas delayed erasers satisfy
\[
d_1(t)=4t-2,
\qquad
d_2(t)=4t-4
\]
for $t\geq 2$. Consequently, sufficiently short quantum messages are
completely hidden from every one- or two-qudit output subsystem, even
when the input is entangled with a reference, while remaining exactly
recoverable from the full output.

For $q=3$, we construct an explicit perfect-tensor Clifford gate from
two inverse SUM gates. Exhaustive Weyl-support searches for
$t=1,2,3$ reproduce the predicted masking distances. For a coherent perturbation of this gate, local leakage scales linearly
with the perturbation strength, whereas the infidelity of recovery using
the ideal inverse scales quadratically near the perfect point.

\end{abstract}

\section{Introduction}

Dual-unitary circuits are built from two-site gates that remain unitary
after the index reshuffling that exchanges the space and time
directions. For traceless one-site observables, their
infinite-temperature two-point correlations vanish away from the two
edges of the causal light cone, while the surviving edge correlations
are generated by iterates of two one-site quantum channels
\(\mathcal{M}_{\pm}\)
\cite[Property~1, Eq.~(12), and Property~2, Eqs.~(17)--(19)]{Bertini2019exact}.
Claeys and Lamacraft subsequently showed that dual-unitary circuits are
maximum-velocity circuits with butterfly velocity \(v_B=1\). For the
maximally chaotic dual-unitary class, they also showed that the
long-time odd-parity out-of-time-order correlator inside the light cone
is determined by the same one-site channels
\cite[Secs.~III and IV.A, Eqs.~(58)--(59)]{ClaeysLamacraft2020maximum}.
Kos, Bertini, and Prosen studied perturbations away from the
dual-unitary point. For a suitable class of unperturbed dual-unitary
gates, they expressed two-point correlations as sums over
one-dimensional transfer matrices contracted along allowed space-time
paths. Their path-sum formula is rigorous in the dilute-defect limit
and for circuits with random longitudinal fields, while the clean
weakly perturbed case is supported by analytical arguments and
numerical tests
\cite[Secs.~III--V, Eqs.~(30)--(33)]{Kos2021Correlations}.
These results motivate a direct study of the light-cone channels. A natural extension relaxes the single-gate dual-unitary
condition itself. Yu, Wang, and Kos introduced a hierarchy of
multi-gate solvability conditions whose first level is ordinary
dual unitarity. At the second level, two-point correlations may
survive both on the light-cone edges and along the time axis, and
the same work gives Weyl-basis constructions of higher-dimensional
dual-unitary and hierarchical gates
\cite[Secs.~2.2--2.4 and 3.1.2]{Yu2024hierarchical}.
Here we classify the corresponding finite-field Clifford dual-unitary
gates at the ordinary dual-unitary level.

Clifford dynamics provides a finite algebraic framework for this
problem. More generally, Clifford conjugation acts linearly on
generalized Pauli labels and preserves their commutation form,
producing a symplectic matrix representation over \(\mathbb{Z}_d\)
\cite[Sec.~II.B--C, Eqs.~(5)--(10)]{Hostens2005Stabilizer}.
In the finite-field model used here, \(n\)-qudit Weyl operators modulo
phase are labelled by vectors in \(\F_q^{2n}\), and their commutation
relations are encoded by the trace-symplectic form
\cite[Sec.~2, Corollary~4, and Lemma~5]{Ketkarnonbinary}.
We restrict to the \(\F_q\)-linear Clifford transformations whose
induced label action lies in \(\Sp(2n,q)\). For a recent overview of the binary-extension Galois-qudit formulation, see \cite[Secs.~2.2 and~3.2]{wills2026review}.

The structural description of Clifford quantum cellular automata as
reversible Clifford dynamics induced by symplectic cellular automata,
including the characterization of their one-dimensional
translationally invariant form and elementary generators, was
developed in \cite{Schlingemann2008on}. This formalism was later used
for qubit spacetime-translation-invariant dual-unitary Clifford
circuits to analyze operator spreading, recurrence times, one-site
gliders, and the performance of circuit-generated codes under erasure
errors
\cite[Secs.~II.B--C, III.C--D, and VI.B--D]{Sommers2023Crystalline}.
For monitored dual-unitary Clifford circuits with probabilistic
measurements chosen to preserve spatial unitarity, Yao and Claeys
characterized temporal-entanglement barriers and found qualitatively
different measurement-modified scaling for good and poor scramblers
\cite[Table~I and Secs.~III--VIII]{Yao2024Temporal}.
More recently, Bejan, Claeys, and Yao studied the spreading of locally
injected magic under random Clifford dynamics and tested the resulting
magic-length-scale phenomenology in identity-doped dual-unitary
Clifford circuits. They found ballistic magic length scales with
velocities approximately \(v_E\) and \(2v_E\), and noted that the
undoped dual-unitary Clifford point is qualitatively distinct
\cite[App.~S8, especially Fig.~S7 and the final paragraph]{Bejan2025Magic}.
These works address operator spreading, recurrence, coding behavior,
temporal entanglement, and magic transport. Here we instead classify
undoped homogeneous finite-field Clifford dual-unitary gates under
ordered local equivalence and determine their exact decoupling
properties for bounded output regions.

Unimodular eigenoperators of the light-cone channels are
maximal-velocity one-site solitons, also called gliders. Holden-Dye,
Masanes, and Pal extended this picture to width-\(w\) solitons and
proved that, under the finite-support assumptions of their theorem,
every conserved quantity formed from finite-range densities decomposes
into charges generated by such solitons. In the infinite-chain limit,
they obtain a one-to-one correspondence between finite-width solitons
and conserved densities
\cite[Sec.~2.3, Theorem~1, and Remark~1]{holden2312fundamental}.
We instead classify the one-site transport channels of finite-field
Clifford dual-unitary gates and use their nilpotent or gliding
structure to derive the transport-phase counts and local masking
distances. Here
\emph{local masking} means exact state-independent reduction on every
output subsystem of the prescribed size.

This notion is distinct from standard bipartite quantum data hiding.
In the foundational formulation of DiVincenzo, Leung, and Terhal, a
classical label is encoded into orthogonal bipartite mixed states that
are perfectly distinguishable by a global measurement but reveal only
a vanishing amount of information to observers restricted to local
operations and classical communication
\cite[Abstract and Secs.~I--III]{divincenzoquantum}.
In the present setting, an observer may perform an arbitrary
measurement on the bounded output subsystem it controls, but has no
access to its complement. The resulting privacy statement is exact
subsystem decoupling rather than approximate indistinguishability under
a restricted measurement class.

Perfect tensors provide a second background connection. For four
parties, the perfect-tensor condition is equivalent to the associated
\(q^2 \times q^2\) coefficient matrix being \(2\)-unitary. Such matrices
form a distinguished subset of dual-unitary gates and, through the
operator--state correspondence, generate \(\operatorname{AME}(4,q)\)
states
\cite[Sec.~VI.B--C and App.~B]{goyeneche2015absolutely}.
See also
\cite[p.~1 and Supplemental Material]{rather2020creating} and
\cite[pp.~2--5, Sec.~II.B, Eq.~(14), and Definitions~1--3]{rather2022construction}.
The projective classification of the four columns is governed by the
classical cross-ratio on \(\operatorname{PG}(1,q)\)
\cite[Sec.~6.1]{Hirschfeld_projective}.
For a recent survey of AME-state constructions, their relation to
\(2\)-unitary matrices and perfect tensors, local-equivalence questions,
and applications to quantum codes and tensor networks, see
\cite[Secs.~II--V, VII--X, and App.~C]{Rajchel2026}.
In the ordered classification below, the cross-ratio parametrizes the
perfect sector.

Minimal-support \(\operatorname{AME}(N,d)\) states are in direct
correspondence with classical maximum-distance-separable codes
\cite[Sec.~IV.B]{goyeneche2015absolutely}
\cite[Sec.~III, especially Eqs.~(4)--(6)]{raissi2018optimal}.
Even-party absolutely maximally entangled states are related to
pure-state threshold quantum secret-sharing schemes through the
correspondence proved in
\cite[Theorem~2]{helwig2012absolute}.
More broadly, multipartite entanglement, quantum error-correcting codes,
and the entangling power of unitary evolutions can be studied within
the common framework of average bipartite entanglement measures
developed in \cite{scott2004multipartite}. In the protocol studied
here, an arbitrary state is placed on a contiguous input interval, while
the remaining sites are initialized in an independent maximally mixed state, and a
translation-invariant brickwork circuit is applied. How large can the
interval be while every observer controlling at most one or two output
sites receives no information about the input?

The mixed exterior supplies the randomness, while the full output
remains reversibly decodable. Thus the construction lies outside the
deterministic bipartite masking model of Modi \emph{et al.}
\cite[Definition~1 and Theorem~3]{modi2018masking}.
Multipartite deterministic masking is nevertheless possible. Li and
Wang constructed encodings of arbitrary \(d\)-level states into
multipartite systems and, whenever a pair of orthogonal Latin squares
of order \(d\) exists, into three \(d\)-level systems whose one-party
marginals are independent of the input state
\cite[Definition~1 and Theorems~1--2]{li2018multipartite}.
The present construction differs because the independent maximally
mixed exterior supplies the randomness and the privacy threshold is
generated dynamically as a function of circuit depth. Its reduced-state
condition is closest to the multipartite \(k\)-uniform masking condition
of Shi \emph{et al.}, which requires every reduction to any \(k\)
output parties to be independent of the encoded state
\cite[Definition~9]{shi2021k}.
Only the reduced-state privacy condition is shared. Here that condition
is generated dynamically for every allowed input state and at every
circuit time covered by the masking law.

The first result is an ordered local classification of isolated gates.
The dual-unitary locus has exactly \(q\) nonlocal cores: \(q-2\) perfect
cores indexed by \(\delta=\det B=\det C\), one rank-one core, and one
SWAP core. The parameter \(\delta\) is a M\"obius reparameterization of the
classical cross-ratio of an ordered \([4,2,3]_q\) MDS configuration. We also compute the stabilizers
and orbit sizes.

The circuit classification is finer. Under homogeneous repetition,
the rank-one core separates into delayed-erasure and gliding phases,
giving five transport phases in total. Their exact counts form a
codimension ladder in \(\Sp(4,q)\). The one-site edge channels determine
the all-time masking laws. Immediate erasure gives
\[
d_1(t)=4t,
\qquad
d_2(t)=4t-2,
\]
while delayed erasure gives
\[
d_1(t)=4t-2,
\qquad
d_2(t)=4t-4,
\qquad
t\geq 2.
\]
Every nonnilpotent edge channel produces a one-site glider. Thus
isolated-gate equivalence is strictly coarser than homogeneous-circuit
transport.

The qubit case is the boundary case because there is no pure
four-qubit state whose two-qubit reduced density matrices are all
maximally mixed. Equivalently, no \(\operatorname{AME}(4,2)\) state
exists
\cite[Sec.~2, Theorem~1]{higuchi2000entangled}.
A modern direct treatment, including seven proofs of this
nonexistence result, is given in
\cite[Observation~1 and Proofs~1--7]{Huber_2026}.
Its optimal phase is delayed erasure, and a signed scalar edge witness
distinguishes delayed erasure, one-edge gliding, two-edge gliding, and
SWAP transport.

Relative to earlier studies of translation-invariant qubit Clifford
circuits, this work adds three elements. First, the classification
applies uniformly to the restricted finite-field Clifford groups over
every prime power \(q\). Second, it separates ordered local gate cores
from the finer transport phases produced by homogeneous repetition.
Third, it converts the edge-channel classification into exact
reference-secure masking distances for one- and two-site output regions.

The remainder of this paper is organized as follows.
Section~\ref{sec:setup} introduces the finite-field Clifford
representation and local masking. Section~\ref{sec:cores} gives the
core classification. Sections~\ref{sec:counting}
and~\ref{sec:concentration} count the transport phases and describe
their geometry. Sections~\ref{sec:masking} and~\ref{sec:privacy} prove
the masking laws and their privacy interpretation. Section~\ref{sec:qutritdemo}
gives the qutrit implementation and analyzes coherent and stochastic
errors. Sections~\ref{sec:qubit} and~\ref{sec:witness} specialize the results to
qubits. Higher-local phenomena and the remaining open problems are
discussed in the appendices.
\section{Finite-field Clifford gates and local masking}
\label{sec:setup}

\subsection{Finite-field phase space}

Let \(q=p^m\) be a prime power, and let \(V=\F_q^2\) carry the
alternating form
\[
[u,v]=u^T Jv,
\qquad
J=
\begin{pmatrix}
0&1\\
-1&0
\end{pmatrix}.
\]

The generalized Pauli and Clifford groups in arbitrary Hilbert-space
dimension \(d\) admit a modular-arithmetic representation over
\(\mathbb Z_d\), in which Clifford conjugation induces a symplectic
transformation of the Pauli labels
\cite[Sec.~II.B--C, Eqs.~(5)--(10)]{Hostens2005Stabilizer}.
We instead use the finite-field, or restricted, Clifford construction.  Let \(q=p^m\), let
\(\omega=\exp(2\pi i/p)\), and let
\(\operatorname{tr}_{q/p}\colon\F_q\to\F_p\) denote the field trace.
The finite-field shift and phase operators are
\[
X(a)\lvert x\rangle=\lvert x+a\rangle,
\qquad
Z(b)\lvert x\rangle
=
\omega^{\operatorname{tr}_{q/p}(bx)}\lvert x\rangle,
\qquad
a,b,x\in\F_q.
\]
The operators \(X(a)Z(b)\) form a finite-field Weyl error basis, so,
modulo phase, a one-qudit Weyl operator is labelled by a pair
\((a,b)\in\F_q^2\).  Their commutation relation is determined by the
trace-symplectic pairing
\[
\operatorname{tr}_{q/p}\!\left(ba'-b'a\right)
\]
\cite[Sec.~2, Eq.~(1), Lemmas~1 and~5]{Ketkarnonbinary}.
In the binary-extension setting this construction is commonly described
using Galois-qudit Pauli operators; for a recent review of that language,
see \cite[Secs.~2.2 and~3.2]{wills2026review}.

We restrict to the \(\F_q\)-linear Clifford transformations.  Under the
standard identification of the label space with a two-dimensional
symplectic vector space over \(\F_q\), the projective restricted
one-qudit Clifford group is
\(\SL(2,q)=\Sp(2,q)\)
\cite[Sec.~II, the two paragraphs following Eq.~(7), and Sec.~III,
the paragraph following Eq.~(11)]{Zhu2017Multiqubit}.  Accordingly, we take the Weyl labels, modulo
phase, to lie in
\[
V=\F_q^2,
\]
with projective restricted Clifford action
\[
\Sp(2,q)=\SL(2,q).
\]
For two qudits, the corresponding symplectic quotient is
\[
\Sp(4,q)
=
\left\{
S\in\GL(4,q):S^T\Omega S=\Omega
\right\},
\qquad
\Omega=J\oplus J.
\]

The matrix realization of the symplectic group and the
low-dimensional identity
\[
\Sp(2,q)=\SL(2,q)
\]
are standard
\cite[Ch.~8, pp.~68--70, especially Theorem~8.1]
{taylor1992geometry}.  For prime \(q\), this is the full projective two-qudit Clifford
quotient.  For nonprime prime powers \(q\), it is the standard
finite-field, or restricted, Clifford subgroup
\cite[Sec.~II, the two paragraphs following Eq.~(7), and the
paragraph following Eq.~(11)]{Zhu2017Multiqubit}.  Throughout the paper, we count symplectic classes, meaning Clifford gates modulo Weyl displacements and overall phases.

Write
\[
 S=\begin{pmatrix}A&B\\C&D\end{pmatrix},
 \qquad
 A,B,C,D\in M_2(\F_q).
\]
The identity
\begin{equation}
 X^TJX=(\det X)J
 \label{eq:detidentity}
\end{equation}
holds for every $2\times2$ matrix over $\F_q$, including characteristic two.

\subsection{Dual-unitarity and perfect tensors}

Let
\[
\lvert\Phi_{q^2}\rangle
=
\frac{1}{q}
\sum_{x,y\in\F_q}
\lvert x,y\rangle\otimes\lvert x,y\rangle
\]
be the normalized maximally entangled state between two copies of the
two-qudit Hilbert space.  This state is a tensor product of two
generalized Bell stabilizer states.  If \(U\) is a two-qudit Clifford
gate, then its normalized Choi state is
\[
\lvert U\rangle
=
(I_{q^2}\otimes U)\lvert\Phi_{q^2}\rangle.
\]
Since Clifford operations map stabilizer states to stabilizer states,
\(\lvert U\rangle\) is a four-qudit stabilizer state
\cite[App.~C, Eq.~(C7)]{Hostens2005Stabilizer}.
Under the standard gate--state reshaping, a two-qudit gate is
\(2\)-unitary precisely when the gate matrix, its partial transpose,
and its reshuffling are all unitary.  Equivalently, its four-index
coefficient tensor defines an \(\operatorname{AME}(4,q)\) state
\cite[Sec.~VI.B--C and App.~B]{goyeneche2015absolutely}.  Such a state
is a four-party perfect tensor: for every bipartition of its indices
into sets \(X\) and \(X^c\) with \(\lvert X\rvert\leq\lvert X^c\rvert\),
the tensor is proportional to an isometry from \(X\) to \(X^c\)
\cite[Sec.~2, Definition~2]{pastawski2015holographic}.
For a bipartite stabilizer state, the entanglement entropy is determined
by the rank of the subgroup of stabilizers supported locally on either
side.  Equivalently, the state can be reduced by local unitaries to a
collection of independent maximally entangled pairs
\cite[Result~1 and Eqs.~(1), (6), and (8)]{fattal2004entanglement}.
Thus, maximal entanglement across a balanced bipartition is
equivalent to the absence of a nontrivial stabilizer supported entirely
on either side.  Applied to the two spacetime bipartitions, this gives
the following block criterion.

\begin{lemma}
\label{lem:blockcriteria}
A finite-field two-qudit Clifford symplectic class is dual-unitary if and only if
\[
 B,C\in\GL(2,q).
\]
It is a perfect tensor, equivalently two-unitary, if and only if all four blocks $A,B,C,D$ are invertible.
\end{lemma}

\begin{proof}
Consider the balanced spacetime bipartition for which the first output
site is grouped with the first input site.  A Weyl label supported on
the second input site has the form \((0,y)\), with \(y\in\F_q^2\), and
under the symplectic action it is mapped to
\[
S(0,y)=(By,Dy).
\]
Its image is supported entirely on the second output site if and only if
\(By=0\).  Hence a nonidentity Weyl operator of this form exists exactly
when \(\ker B\neq\{0\}\).  In the Choi stabilizer state, such an operator
gives a nontrivial stabilizer supported wholly on one side of the
corresponding balanced bipartition.  By the stabilizer entropy
criterion, that bipartition is maximally entangled if and only if no
such nontrivial local stabilizer exists.  Therefore maximal
entanglement across this cut is equivalent to \(B\in\GL(2,q)\).

Applying the same argument to a Weyl label supported on the first input
site shows that maximal entanglement across the opposite spacetime cut
is equivalent to \(C\in\GL(2,q)\).  Thus the gate is dual-unitary if and
only if
\[
B,C\in\GL(2,q).
\]

For the remaining inequivalent balanced bipartition, the corresponding
local stabilizers are parametrized by \(\ker A\) and \(\ker D\).
Maximal entanglement across this cut is therefore equivalent to
\(A,D\in\GL(2,q)\).  Hence the Choi state is maximally entangled across
all three balanced bipartitions, and thus is a perfect tensor, if and
only if all four blocks \(A,B,C,D\) are invertible.
\end{proof}

\begin{corollary}
\label{cor:blockentropies}
Let \(\lvert S\rangle\) be the normalized four-qudit Choi state
associated with
\[
S=\begin{pmatrix}A&B\\C&D\end{pmatrix}\in\Sp(4,q),
\]
and measure entropy in \(q\)-ary units,
\[
\mathsf E_q(X|X^c)
=
\frac{S(\rho_X)}{\log q}.
\]
For the two nontrivial balanced spacetime cuts, the entropies are
\[
\mathsf E_q^{(B)}=\rank B=\rank C,
\qquad
\mathsf E_q^{(A)}=\rank A=\rank D.
\]
The input--output bipartition always has entropy \(2\).  The three ordered local core types have balanced-cut entropy profiles,
up to the ordering of the cuts,
\[
S_\delta:\ (2,2,2),
\qquad
S_{\mathsf R}:\ (2,2,1),
\qquad
S_{\mathsf S}:\ (2,2,0).
\]
\end{corollary}

\begin{proof}
For a balanced bipartition with two qudits on each side, the
finite-field stabilizer entropy formula gives
\[
\mathsf E_q(X|X^c)
=
2-\dim_{\F_q}\mathcal S_X,
\]
where \(\mathcal S_X\) is the subgroup of Choi stabilizers supported
entirely on \(X\)
\cite[Result~1 and Eqs.~(6) and~(8)]{fattal2004entanglement}.
For the spacetime cut detected by \(B\), these local stabilizers are
parametrized by labels \(v\in\ker B\).  Hence
\[
\dim_{\F_q}\mathcal S_X
=
\dim_{\F_q}\ker B
=
2-\rank B,
\]
which gives \(\mathsf E_q^{(B)}=\rank B\).  Applying the same argument
to the opposite side gives \(\mathsf E_q^{(B)}=\rank C\).  The remaining
balanced spacetime cut similarly gives
\(\mathsf E_q^{(A)}=\rank A=\rank D\).  Finally, the normalized Choi vector of a unitary operator is obtained
from the canonical maximally entangled vector by applying the unitary
to one tensor factor.  Hence the input--output cut is maximally
entangled and has entropy \(2\) in \(q\)-ary units
\cite[Example~2.8, pp.~68--69, and Sec.~2.2.3, pp.~91--92,
Eq.~(2.143)]{watrous2018theory}.
\end{proof}

\subsection{Brickwork circuits and masking distances}

Let $C$ be a two-qudit Clifford gate.  One Floquet period on the infinite chain is
\begin{equation}
 \cU_C=
 \left(\prod_{n\in\mathbb Z}C_{2n+1,2n+2}\right)
 \left(\prod_{n\in\mathbb Z}C_{2n,2n+1}\right),
 \label{eq:brickwork}
\end{equation}
where the right layer acts first.  A finite Weyl string $P$ has weight $\wt(P)$ and contiguous span $\ell(P)$, the length of the smallest interval containing its support.

\begin{definition}
For $r\ge1$ and $t\ge0$, define
\begin{equation}
 d_r^C(t)=
 \min_{\substack{P\ne I\\
 \wt(\cU_C^tP\cU_C^{-t})\le r}}
 \ell(P).
 \label{eq:dr}
\end{equation}
Equivalently, $d_r^C(t)$ is the smallest span of the backward image of a nonidentity output Weyl string of weight at most $r$.
\end{definition}

\begin{theorem}
\label{thm:referenceprivacy}
Work on any finite chain containing the backward causal cone of the observed output set.  Let $M$ be a contiguous input interval and define the mixed-state encoding
\begin{equation}
 \cE_{t,M}(\rho_M)
 =
 \cU_C^t
 \left(
 \rho_M\otimes\frac{I_{M^c}}{q^{|M^c|}}
 \right)
 \cU_C^{-t}.
 \label{eq:privateencoder}
\end{equation}
For an output set $A$, let
\[
 \cN_{A\leftarrow M}^{(t)}
 =\tr_{A^c}\circ\cE_{t,M}.
\]
If $|M|<d_r^C(t)$ and $|A|\le r$, then $\cN_{A\leftarrow M}^{(t)}$ is completely depolarizing,
\begin{equation}
 \cN_{A\leftarrow M}^{(t)}(X)
 =\tr(X)\frac{I_A}{q^{|A|}}
 \qquad\text{for every operator }X\text{ on }M.
 \label{eq:depolarizingobserver}
\end{equation}
Equivalently, for every reference system $R$ and every state $\rho_{RM}$,
\begin{equation}
 (\id_R\otimes\cN_{A\leftarrow M}^{(t)})(\rho_{RM})
 =
 \rho_R\otimes\frac{I_A}{q^{|A|}}.
 \label{eq:referenceprivacy}
\end{equation}
The set $A$ may be disconnected.
\end{theorem}

\begin{proof}
It is enough to test the reduced state against operators $X_R\otimes W_A$, where $X_R$ is arbitrary and $W_A$ ranges over the Weyl basis on $A$.  The identity $W_A=I_A$ gives the correct marginal $\rho_R$.  Suppose $W_A\ne I_A$.  Its weight is at most $r$, so the backward image
\[
 \cU_C^{-t}W_A\cU_C^t
\]
has contiguous span at least $d_r^C(t)$.  Since $|M|<d_r^C(t)$, this Weyl string cannot be supported entirely in $M$.  It acts nontrivially on at least one site of the maximally mixed pad $M^c$, and its trace against that pad is therefore zero.  Hence every nonidentity Weyl coefficient on $A$ vanishes, independently of $X_R$ and $\rho_{RM}$.  This proves \eqref{eq:referenceprivacy}, and the operator identity \eqref{eq:depolarizingobserver} follows by linearity.
\end{proof}

\begin{corollary}
\label{cor:fulldecoder}
The complete output permits exact recovery of the input state.  Applying $\cU_C^{-t}$ to \eqref{eq:privateencoder} returns
\[
 \rho_M\otimes\frac{I_{M^c}}{q^{|M^c|}},
\]
so tracing out $M^c$ recovers $\rho_M$ exactly.  No purification of the pad is required by the decoder.
\end{corollary}

\begin{remark}
The maximally mixed pad must be independent of the message.  If it is prepared by purification into an auxiliary system $E$, the theorem assumes that $E$ is discarded, retained by the encoder, or otherwise inaccessible to the observer.  The claim does not cover an adversary who receives both $A$ and the pad purification.
\end{remark}

A depth-$2t$ nearest-neighbor circuit gives the elementary causal ceilings
\begin{equation}
 d_1(t)\le4t,
 \qquad
 d_2(t)\le4t-2.
 \label{eq:causalceilings}
\end{equation}
The second bound is obtained by choosing a two-site output that is the image, under the final two-qudit gate, of a one-site Weyl operator.

\section{Ordered local-core classification}
\label{sec:cores}

The local classification of a single gate is coarser than the transport classification of the repeated circuit.  Under ordered local finite-field Clifford transformations, the five transport phases reduce to three core types.  The different rank-one transport phases arise from the way the local factors are repeated in the homogeneous wiring.

\begin{definition}
Let
\[
 \cL_q=\SL(2,q)^4.
\]
For $P,Q,R,T\in\SL(2,q)$ define
\begin{equation}
 (P,Q;R,T)\cdot S
 =
 \begin{pmatrix}P&0\\0&Q\end{pmatrix}
 S
 \begin{pmatrix}R&0\\0&T\end{pmatrix}^{-1}.
 \label{eq:localaction}
\end{equation}
This is the symplectic quotient of post- and pre-composition by ordered one-qudit finite-field Clifford gates.
\end{definition}

Let
\[
 \cD_q=
 \left\{
 S=\begin{pmatrix}A&B\\C&D\end{pmatrix}\in\Sp(4,q):
 B,C\in\GL(2,q)
 \right\}
\]
be the dual-unitary locus.  For $a\in\F_q^\times$ put
\[
 K_a=\begin{pmatrix}a&0\\0&1\end{pmatrix},
 \qquad
 E=\begin{pmatrix}1&0\\0&0\end{pmatrix},
 \qquad
 F=JEJ=\begin{pmatrix}0&0\\0&-1\end{pmatrix}.
\]
For $\delta\in\F_q\setminus\{0,1\}$ define
\begin{equation}
 S_\delta=
 \begin{pmatrix}
 K_{1-\delta}&K_\delta\\[1mm]
 K_\delta&D_\delta
 \end{pmatrix},
 \qquad
 D_\delta=
 \begin{pmatrix}
 -\delta&0\\[1mm]
 0&\dfrac{\delta-1}{\delta}
 \end{pmatrix},
 \label{eq:canonicalcore}
\end{equation}
while
\begin{equation}
 S_{\mathsf R}=
 \begin{pmatrix}E&I\\ I&F\end{pmatrix},
 \qquad
 S_{\mathsf S}=
 \begin{pmatrix}0&I\\ I&0\end{pmatrix}.
 \label{eq:rankoneswapcores}
\end{equation}
All these matrices are symplectic.  The first family consists of perfect tensors, $S_{\mathsf R}$ has rank pattern $(1,2,2,1)$, and $S_{\mathsf S}$ has rank pattern $(0,2,2,0)$.

\begin{theorem}
\label{thm:fullcores}
The ordered local quotient of the dual-unitary locus is
\begin{equation}
 \boxed{
 \cD_q/\cL_q
 \cong
 \bigl(\F_q\setminus\{0,1\}\bigr)
 \sqcup\{\mathsf R,\mathsf S\}.
 }
 \label{eq:qcores}
\end{equation}
Equivalently, there are exactly $q$ ordered local finite-field Clifford cores:
\begin{enumerate}[label=(\roman*)]
\item one perfect core $S_\delta$ for each $\delta\in\F_q\setminus\{0,1\}$;
\item one rank-one core $S_{\mathsf R}$ containing every class with $A\ne0$ singular;
\item one SWAP core $S_{\mathsf S}$ containing every class with $A=0$.
\end{enumerate}
The three types are disjoint because the rank of $A$ is invariant under \eqref{eq:localaction}; within the perfect type, $\delta=\det B=\det C$ is a complete invariant.
\end{theorem}

\begin{proof}
Under \eqref{eq:localaction}, $A$ transforms to $PAR^{-1}$ and $B$ to $PBT^{-1}$.  Hence $\rank A$ and $\det B$ are invariant.

Suppose first that $A$ is invertible.  By the determinant parameterization proved in \Cref{lem:detparam},
\[
 \det A=1-\delta,
 \qquad
 \det B=\det C=\delta,
\]
with $\delta\notin\{0,1\}$.  Set
\[
 P=K_{1-\delta}A^{-1},
 \qquad R=I,
 \qquad T=K_\delta^{-1}PB,
 \qquad Q=K_\delta C^{-1}.
\]
All four matrices lie in $\SL(2,q)$, and the transformed first three blocks are $K_{1-\delta},K_\delta,K_\delta$.  The symplectic equations uniquely determine the fourth block, giving $S_\delta$.  Distinct values of $\delta$ cannot be equivalent.

Now assume $A$ is nonzero singular.  Then $\delta=1$, and $A$ and $D$ both have rank one.  The left-right action of $\SL(2,q)\times\SL(2,q)$ is transitive on nonzero rank-one matrices, so choose $P,R$ with
\[
 PAR^{-1}=E.
\]
Set
\[
 T=PB,
 \qquad
 Q=RC^{-1}.
\]
Since $\det B=\det C=1$, both $T$ and $Q$ lie in $\SL(2,q)$, and the transformed cross blocks are the identity.  Symplecticity then forces the fourth block to be $F=JEJ$.  Thus every nonzero rank-one class is equivalent to $S_{\mathsf R}$.

Finally, if $A=0$, then the off-diagonal symplectic equation and invertibility of $B,C$ force $D=0$.  The same choices $T=PB$ and $Q=RC^{-1}$ normalize the two cross blocks to the identity, giving $S_{\mathsf S}$.
\end{proof}

\begin{theorem} 
\label{thm:allstabilizers}
Let
\[
h=|\SL(2,q)|=q(q^2-1),
\]
where the equality follows from the standard order formula for
\(\SL(m,q)\)
\cite[Ch.~4, p.~19]{taylor1992geometry}.
The three core types have the following stabilizers and ordered
local-orbit sizes.

\[
\begin{array}{c|c|c|c}
\text{core}&\text{number of core parameters}&|\operatorname{Stab}|&|\operatorname{Orb}|\\
\hline
S_\delta&q-2&h&h^3\\
S_{\mathsf R}&1&q^2(q-1)&h^2(q-1)(q+1)^2\\
S_{\mathsf S}&1&h^2&h^2.
\end{array}
\]
Summing the ordered local orbits reproduces the full dual-unitary count
\[
 (q-2)h^3+h^2(q-1)(q+1)^2+h^2
 =h^2q(q^3-q^2+1).
\]
\end{theorem}

\begin{proof}
For $S_\delta$, the stabilizer is parametrized by $P\in\SL(2,q)$ through
\[
 R=K_{1-\delta}^{-1}PK_{1-\delta},
 \qquad
 T=K_\delta^{-1}PK_\delta,
 \qquad
 Q=K_\delta RK_\delta^{-1},
\]
so it has size $h$.

For $S_{\mathsf R}$, the block equations imply $T=P$, $Q=R$, and
\[
 PE=ER,
 \qquad
 RF=FP.
\]
Writing the matrices explicitly gives
\[
 P=\begin{pmatrix}a&b\\0&a^{-1}\end{pmatrix},
 \qquad
 R=\begin{pmatrix}a&0\\u&a^{-1}\end{pmatrix},
\]
with $a\in\F_q^\times$ and $b,u\in\F_q$.  Thus the stabilizer has size $q^2(q-1)$.

For $S_{\mathsf S}$, the only conditions are $T=P$ and $Q=R$, with $P,R\in\SL(2,q)$ arbitrary, so the stabilizer has size $h^2$.  Orbit--stabilizer gives the stated orbit sizes.
\end{proof}

\begin{corollary}
\label{cor:corevsphase}
The perfect cores $S_\delta$ form the immediate-erasure stratum, and the SWAP core forms the SWAP-transport stratum.  The unique rank-one core $S_{\mathsf R}$ contains all four rank-one edge-incidence patterns: delayed bilateral erasure, left-only gliding, right-only gliding, and two-edge gliding.  Therefore these four wiring patterns are not invariants of an isolated dual-unitary gate under ordered local Clifford equivalence.
\end{corollary}

\begin{corollary}
\label{cor:compilation}
For perfect classes,
\[
 S\sim_{\mathrm{ordLC}}S'
 \iff
 \det B=\det B'.
\]
Every dual-unitary symplectic class admits a local-plus-core decomposition using exactly one of $S_\delta,S_{\mathsf R},S_{\mathsf S}$.  Uniform sampling from any fixed ordered local orbit is obtained by choosing $P,Q,R,T$ independently and uniformly from $\SL(2,q)$ and applying \eqref{eq:localaction}; every orbit element has exactly the stabilizer size shown in \Cref{thm:allstabilizers} preimages.
\end{corollary}

\subsection{MDS cross-ratio of the perfect cores}
The canonical perfect core is diagonal in the \(x\)- and \(z\)-phase-space coordinates. Its \(x\)-coordinate graph is the classical \([4,2,3]_q\) MDS code with generator matrix
\begin{equation}
G_{\delta}
=
\begin{pmatrix}
1 & 0 & 1-\delta & \delta \\
0 & 1 & \delta & -\delta
\end{pmatrix}.
\label{eq:MDSgenerator}
\end{equation}
The connection between minimal-support absolutely maximally entangled states and classical MDS codes is described in \cite[Sec.~IV.B]{goyeneche2015absolutely}.

If \(c_1,\ldots,c_4\) are its projective columns, their ordered
cross-ratio is
\begin{equation}
\lambda
=
\frac{\det(c_1,c_3)\det(c_2,c_4)}
{\det(c_1,c_4)\det(c_2,c_3)}
=
\frac{\delta}{\delta-1}.
\label{eq:crossratio}
\end{equation}
For four ordered distinct points of \(\operatorname{PG}(1,q)\),
equality of the corresponding cross-ratios is equivalent to
projective equivalence in the prescribed order
\cite[Sec.~6.1]{Hirschfeld_projective}.
The map
\[
\delta\longmapsto\frac{\delta}{\delta-1}
\]
is an involution of \(\F_q\setminus\{0,1\}\).  The
\(q-2\) ordered perfect cores recover exactly the classical
cross-ratio moduli of four ordered projective points.  This projective modulus labels the ordered Clifford core and already
appears in the higher-local transport of canonical homogeneous
circuits.

If the four tensor legs are not ordered, the action of \(S_4\) on
their orderings descends through the Klein four subgroup that fixes
the cross-ratio.  The resulting quotient
\[
S_4/V_4\cong S_3
\]
is the classical anharmonic group.  Its action identifies the six
cross-ratio representatives
\begin{equation}
\lambda,
\qquad
1-\lambda,
\qquad
\lambda^{-1},
\qquad
(1-\lambda)^{-1},
\qquad
\frac{\lambda}{\lambda-1},
\qquad
\frac{\lambda-1}{\lambda}.
\label{eq:anharmonicorbit}
\end{equation}
This permutation action and its six cross-ratio values are reviewed
in \cite[Sec.~6.1]{Hirschfeld_projective}.
Let $\epsilon_2=0$ in characteristic two and $\epsilon_2=1$ otherwise, and let
\[
 \rho_q=\#\{x\in\F_q:x^2-x+1=0\}.
\]

The anharmonic group has one identity element, three involutions,
and two elements of order three.  The identity fixes all \(q-2\)
elements of \(\F_q\setminus\{0,1\}\).  The three involutions may be
taken as
\[
 \lambda\longmapsto1-\lambda,
 \qquad
 \lambda\longmapsto\lambda^{-1},
 \qquad
 \lambda\longmapsto\frac{\lambda}{\lambda-1}.
\]
Their fixed points are respectively
\[
\frac12,\qquad -1,\qquad 2.
\]
These are precisely the three ordered cross-ratio representatives of
a harmonic tetrad
\cite[Sec.~6.1]{Hirschfeld_projective}.
They belong to \(\F_q\setminus\{0,1\}\) precisely when
\(\operatorname{char}\F_q\ne2\).  Thus each involution has
\(\epsilon_2\) fixed points.  In characteristic three the three
representatives coincide, since
\[
-1=2=\frac12.
\]  The two order-three transformations
may be written as
\[
 \lambda\longmapsto\frac{1}{1-\lambda},
 \qquad
 \lambda\longmapsto\frac{\lambda-1}{\lambda}.
\]
For either transformation, the fixed-point equation is
\[
 \lambda^2-\lambda+1=0,
\]
so each has \(\rho_q\) fixed points.  Burnside's lemma therefore gives
\begin{equation}
 N_{\rm perm}(q)
 =
 \frac{(q-2)+3\epsilon_2+2\rho_q}{6}.
 \label{eq:permutationorbits}
\end{equation}

\section{Five transport phases}
\label{sec:counting}

Put
\[
 h=|\SL(2,q)|=q(q^2-1),
 \qquad
 L=h^2.
\]

\begin{lemma}
\label{lem:detparam}
Let $S$ be dual-unitary.  Then there is a unique $\delta\in\F_q^\times$ such that
\[
 \det B=\det C=\delta,
 \qquad
 \det A=\det D=1-\delta.
\]
Conversely, choose $B,C$ with determinant $\delta\ne0$ and choose $A$ with determinant $1-\delta$.  Then
\begin{equation}
 D=-(C^TJ)^{-1}A^TJB
 \label{eq:Dformula}
\end{equation}
is the unique block for which $S$ is symplectic.
\end{lemma}

\begin{proof}
The diagonal symplectic equations and \eqref{eq:detidentity} give
\[
 \det A+\det C=1,
 \qquad
 \det B+\det D=1.
\]
The off-diagonal equation is
\[
 A^TJB+C^TJD=0,
\]
which yields \eqref{eq:Dformula}.  Taking determinants gives
\[
 \det D=\frac{\det A\det B}{\det C}.
\]
Substituting the two diagonal equations and using $\det B,\det C\ne0$ gives $\det B=\det C$.  The converse follows by reversing the argument.
\end{proof}

The determinant fiber
\[
 \{X\in M_2(\F_q):\det X=a\}
\]
has cardinality $h$ for every $a\ne0$. 

The one-site edge channels used below are the finite-field Clifford and Weyl-basis specialization of the light-cone maps \(\mathcal M_\pm\) introduced by Bertini, Kos, and Prosen \cite[Property~2, Eqs.~(17)--(19)]{Bertini2019exact}. Eigenoperators of these maps with eigenvalues of unit modulus propagate along a light-cone edge without spreading and are therefore maximal-velocity one-site solitons, or gliders, in the terminology of Holden-Dye, Masanes, and Pal \cite[Sec.~2.3, Eqs.~(14)--(17), and Definition~1]{holden2312fundamental}. The same channels determine the long-time odd-parity out-of-time-order correlator inside the light cone for the maximally chaotic dual-unitary class \cite[Eqs.~(58)--(59)]{ClaeysLamacraft2020maximum}.

\begin{definition}
Let $W_x$, $x\in\F_q^2$, denote the Weyl basis with $W_0=I$.  The right and left one-site edge channels are the partial-trace superoperators obtained from one backward gate application with the inward leg maximally mixed.  In the normalized nonidentity Weyl basis $\{W_x/\sqrt q:x\ne0\}$, write their transfer matrices as $T_R$ and $T_L$.  For a Clifford gate, each column of these matrices is either zero or a single phase times another Weyl basis vector, so $T_R$ and $T_L$ are partial permutation matrices.  Their Frobenius norms and powers are therefore ordinary linear-superoperator quantities, not norms of a piecewise map on the label space.
\end{definition}

\begin{lemma}
\label{lem:edgemaps}
For a $2\times2$ matrix write
\[
 X^\#=-JX^TJ=\operatorname{adj}(X).
\]
If
\[
 S^{-1}=\begin{pmatrix}A^\#&C^\#\\B^\#&D^\#\end{pmatrix},
\]
then, up to Weyl phases, their action on Weyl basis elements is
\[
 \mathcal T_R(W_x)=
 \begin{cases}
 W_{B^\#x},&A^\#x=0,\\
 0,&A^\#x\ne0,
 \end{cases}
 \qquad
 \mathcal T_L(W_x)=
 \begin{cases}
 W_{C^\#x},&D^\#x=0,\\
 0,&D^\#x\ne0.
 \end{cases}
\]
For a dual-unitary class with determinant parameter $\delta$,
\[
 B^\#=\delta B^{-1},\qquad C^\#=\delta C^{-1}.
\]
The scalar $\delta$ is irrelevant projectively but is required for the exact linear Weyl-label map.

Thus,
\begin{enumerate}[label=(\roman*)]
\item if $A$ and $D$ are invertible, both edge channels vanish;
\item if $A$ and $D$ have rank one, then
\[
 \ker A^\#=\operatorname{im}A,
 \qquad
 \ker D^\#=\operatorname{im}D,
\]
so each edge channel is supported on one projective line and its square is nonzero precisely when the image line returns to the same domain line;
\item if $A=D=0$, then $\ker A^\#=\ker D^\#=\F_q^2$ and both channels are full Weyl permutations.
\end{enumerate}
\end{lemma}

\begin{proof}
Backward conjugation of a right-edge one-site Weyl label $(x,0)$ gives
\[
 S^{-1}(x,0)=(A^\#x,B^\#x).
\]
The partial trace over the inward site survives exactly when $A^\#x=0$, and the surviving exterior label is $B^\#x$.  This proves the right formula; the left formula is identical.  Since $B$ and $C$ are invertible and have determinant $\delta$,
\[
 B^\#=(\det B)B^{-1}=\delta B^{-1},
 \qquad
 C^\#=\delta C^{-1}.
\]
The three rank cases follow from the elementary identities for the adjugate of a $2\times2$ matrix.  In the rank-one case, a second edge step survives exactly when the first image lies again in the one-dimensional domain.
\end{proof}

\begin{theorem}
\label{thm:fivephase}
The number of dual-unitary finite-field Clifford symplectic classes is
\begin{equation}
 N_{\rm DU}(q)=Lq(q^3-q^2+1).
 \label{eq:NDu}
\end{equation}
They split into five edge-transport phases as follows.
\begin{table}[t]
\centering
\begin{tabular}{p{0.58\linewidth}r}
\toprule
\textbf{phase}&\textbf{number of classes}\\
\midrule
Immediate bilateral erasure: $T_L=T_R=0$
&$Lq(q-2)(q^2-1)$\\
Delayed bilateral erasure: $T_L,T_R\ne0$ and $T_L^2=T_R^2=0$
&$Lq^2(q-1)$\\
Exactly one gliding edge
&$2Lq(q-1)$\\
Two rank-one gliding edges
&$L(q-1)$\\
SWAP sector
&$L$\\
\bottomrule
\end{tabular}
\caption{Exact phase counts for finite-field Clifford dual-unitary symplectic classes.}
\label{tab:phasecounts}
\end{table}
The immediate-erasure phase is exactly the perfect-tensor phase.
\end{theorem}

\begin{proof}
If $\delta\ne1$, all four blocks have nonzero determinant.  There are $q-2$ choices of $\delta$ and $h$ choices for each of $A,B,C$, while $D$ is fixed by \eqref{eq:Dformula}.  This gives
\[
 (q-2)h^3=Lq(q-2)(q^2-1)
\]
perfect tensors.

Now let $\delta=1$.  Then $A$ and $D$ are singular.  If $A=0$, equation \eqref{eq:Dformula} gives $D=0$, while $B,C\in\SL(2,q)$ are arbitrary.  These are the $L$ SWAP-sector classes.

Every remaining singular $A$ has rank one and can be written
\[
 A=u v^T.
\]
It is determined by the two projective lines $[u],[v]\in\mathbb P^1(\F_q)$ and one nonzero scalar.  There are therefore $(q+1)^2(q-1)$ possible rank-one matrices.  In the determinant-one sector, \eqref{eq:Dformula} becomes
\[
 D=CJv\,u^TJB.
\]
Thus $\operatorname{im}A=[u]$ and $\operatorname{im}D=CJ[v]$.  By \Cref{lem:edgemaps}, the right edge glides exactly when $B$ fixes $[u]$.  The left condition is
\[
 C(CJ[v])=CJ[v],
\]
projectively, which is equivalent to $C$ fixing $J[v]$.  Hence the two edge incidences depend independently on the pairs $(B,[u])$ and $(C,J[v])$.

The projective line \(\mathbb P^1(\F_q)\) has \(q+1\) points, and the
projective action of \(\SL(2,q)\) is transitive
\cite[Exercise~3.6 and Theorem~4.1]{taylor1992geometry}.
Therefore orbit--stabilizer gives, for every projective line \(\ell\),
\[
 \left|\operatorname{Stab}_{\SL(2,q)}(\ell)\right|
 =
 \frac{|\SL(2,q)|}{q+1}
 =
 q(q-1).
\]
The number of matrix--line pairs \((G,\ell)\) satisfying
\(G\ell=\ell\) is
\[
 (q+1)q(q-1)=h,
\]
whereas the number of nonfixed pairs is \(hq\).
  Multiplying fixed or nonfixed pair counts and the $q-1$ rank-one scalars gives
\[
 Lq^2(q-1),\quad 2Lq(q-1),\quad L(q-1)
\]
for zero, one, or two gliding edges.  Adding the five rows yields \eqref{eq:NDu}.
\end{proof}

\begin{definition}
Using the normalized nonidentity Weyl basis, define
\[
 \Gamma_q(C)=\Frob{T_L(C)^2}^2+\Frob{T_R(C)^2}^2.
\]
The bilateral edge-erasure depth is
\[
 \kappa(C)=\min\{k\ge1:T_L(C)^k=T_R(C)^k=0\},
\]
with $\kappa=\infty$ when no such $k$ exists.
\end{definition}

\begin{corollary}
\label{cor:gammaq}
For dual-unitary finite-field Clifford classes,
\[
 \Gamma_q\in\{0,q-1,2(q-1),2(q^2-1)\}.
\]
The value $\Gamma_q=0$ contains two phases,
\[
 \kappa=1\quad\Longleftrightarrow\quad T_L=T_R=0,
\]
\[
 \kappa=2\quad\Longleftrightarrow\quad T_L,T_R\ne0,
 \quad T_L^2=T_R^2=0.
\]
For $q=2$ the $\kappa=1$ phase is empty.
\end{corollary}

\begin{proof}
A nonzero rank-one edge channel acts on the $q-1$ nontrivial Weyl labels of one projective line.  Its squared Frobenius norm is therefore $q-1$ when it glides and zero when it is nilpotent.  A SWAP edge permutes all $q^2-1$ nonidentity Weyl labels, giving squared Frobenius norm $q^2-1$ per edge.
\end{proof}

For $q=2$, \Cref{thm:fivephase} reduces to
\[
 0,144,144,36,36,
\]
where the initial zero records the absent perfect phase.  After merging the two one-edge orientations, the usual glider-index distribution is
\[
 \#\{\Gamma=0,1,2,6\}=144,144,36,36.
\]
For $q=3$ the five counts are
\[
 13824,\quad10368,\quad6912,\quad1152,\quad576,
\]
with total $32832$.

\section{Large-\texorpdfstring{$q$}{q} behavior}
\label{sec:concentration}

The standard symplectic order formula is
\[
 |\Sp(2m,q)|
 =
 q^{m^2}\prod_{i=1}^{m}(q^{2i}-1).
\]
In particular,
\[
 |\Sp(4,q)|
 =
 q^4(q^2-1)(q^4-1)
\]
\cite[Ch.~8, p.~70]{taylor1992geometry}.

\begin{theorem}
\label{thm:concentration}
For a uniformly random symplectic class $S\in\Sp(4,q)$,
\begin{align}
 \Pr(S\text{ is dual-unitary})
 &=\frac{q^3-q^2+1}{q(q^2+1)},
 \label{eq:probdu}\\
 \Pr(S\text{ is a perfect tensor})
 &=\frac{(q-2)(q^2-1)}{q(q^2+1)}
 =1-\frac2q+O(q^{-2}),
 \label{eq:probperfect}\\
 \Pr(S\text{ has a persistent one-site glider})
 &=\frac{2q-1}{q(q^2+1)}=O(q^{-2}).
 \label{eq:probgliderall}
\end{align}
Conditioned on dual-unitarity,
\begin{align}
 \Pr(\kappa=1\mid\mathrm{DU})
 &=\frac{(q-2)(q^2-1)}{q^3-q^2+1},\\
 \Pr(\kappa=\infty\mid\mathrm{DU})
 &=\frac{2q-1}{q^3-q^2+1}.
\end{align}
Thus perfect tensors dominate both the dual-unitary ensemble and the full finite-field Clifford symplectic ensemble as $q\to\infty$.
\end{theorem}

\begin{proof}
Divide the counts in \Cref{thm:fivephase} by $|\Sp(4,q)|$ and use $L=q^2(q^2-1)^2$.  The persistent-glider count is the sum of the last three nonperfect rows,
\[
 L\bigl(2q(q-1)+(q-1)+1\bigr)=Lq(2q-1).
\]
The asymptotic expansions are immediate.
\end{proof}

\begin{corollary}
\label{cor:rarity}
For a uniformly random $S\in\Sp(4,q)$, the five phase probabilities are
\[
\begin{array}{c|c|c}
\text{phase}&\text{exact probability}&\text{large-$q$ order}\\
\hline
\kappa=1\text{ perfect tensor}
&\dfrac{(q-2)(q^2-1)}{q(q^2+1)}&1-2q^{-1}+O(q^{-2})\\[2mm]
\kappa=2\text{ delayed erasure}
&\dfrac{q-1}{q^2+1}&q^{-1}+O(q^{-2})\\[2mm]
\text{one gliding edge}
&\dfrac{2(q-1)}{q(q^2+1)}&2q^{-2}+O(q^{-3})\\[2mm]
\text{two rank-one gliding edges}
&\dfrac{q-1}{q^2(q^2+1)}&q^{-3}+O(q^{-4})\\[2mm]
\text{SWAP sector}
&\dfrac{1}{q^2(q^2+1)}&q^{-4}+O(q^{-6}).
\end{array}
\]
The five probabilities follow the codimensions of the corresponding algebraic loci in $\Sp_4$.  Each additional edge-incidence condition contributes one inverse power of the local dimension.
\end{corollary}

\begin{theorem}
\label{thm:codimension}
Inside the algebraic group $\Sp_4$, the perfect-tensor locus
\[
 \det A\,\det B\,\det C\,\det D\ne0
\]
is a nonempty Zariski-open dense subset.  The five transport loci in the dual-unitary locus have dimensions and codimensions
\[
\begin{array}{c|cc}
\text{phase}&\dim&\operatorname{codim}_{\Sp_4}\\
\hline
\kappa=1\text{ perfect}&10&0\\
\kappa=2\text{ delayed}&9&1\\
\text{one gliding edge}&8&2\\
\text{two gliding edges}&7&3\\
\text{SWAP}&6&4.
\end{array}
\]

Thus the probability hierarchy in \Cref{cor:rarity} follows from an exact hierarchy of incidence codimensions over finite fields.

\end{theorem}

\begin{proof}
We work over an algebraic closure \(K\) of the ground field.  The
symplectic group \(\Sp_4(K)\) is the closed algebraic subgroup of
\(\GL_4(K)\) defined by the polynomial equations
\[
 S^{T}\Omega S=\Omega
\]
\cite[Sec.~7.2, p.~52]{humphreys2012linear}.

For each \(\delta\in K^\times\), the determinant fiber
\[
 \{X\in M_2(K):\det X=\delta\}
\]
is a left coset of \(\SL_2(K)\).  It therefore has dimension
\[
 \dim\SL_2(K)=2^2-1=3
\]
\cite[Sec.~7.2, p.~52]{humphreys2012linear}.
The perfect parameterization has one nonzero determinant parameter
and three such fibers, hence dimension
\[
 1+3+3+3=10.
\]

The rank-one condition on \(A\) is determinantal:
\[
\rank A\le 1
\quad\Longleftrightarrow\quad
\det A=0.
\]
The variety
\[
\mathcal R_1
=
\{A\in M_2(K):\rank A\le1\}
\]
is irreducible of dimension
\[
2\cdot1+2\cdot1-1^2=3,
\]
and is nonsingular away from its lower-rank locus
\(\mathcal R_0=\{0\}\)
\cite[Sec.~1.C, Prop.~(1.1), pp.~4--5]{BrunsVetter1988}.
Hence,
\[
\mathcal R_1^\times
=
\{A\in M_2(K):\rank A=1\}
\]
is a smooth irreducible locally closed variety of dimension \(3\).
Since \(B,C\in\SL_2(K)\) contribute three dimensions each and
\(D\) is uniquely determined by \eqref{eq:Dformula}, the full
rank-one sector has dimension
\[
3+3+3=9.
\]

Write \(A=uv^T\), with \(u,v\ne0\).  By the edge-incidence
description in the proof of \Cref{thm:fivephase}, the right edge
glides precisely when
\[
B[u]=[u],
\]
and the left edge glides precisely when
\[
C[Jv]=[Jv].
\]
For a fixed projective line \(\ell\subset K^2\), its stabilizer in
\(\SL_2(K)\) is conjugate to
\[
\left\{
\begin{pmatrix}
a&b\\
0&a^{-1}
\end{pmatrix}
:
a\in K^\times,\ b\in K
\right\},
\]
which has dimension \(2\).  Thus each independent fixed-line
incidence lowers the dimension by one.  It follows that the delayed
bilateral-erasure locus, where neither incidence holds, is a
nonempty open subset of the rank-one sector and has dimension \(9\).
The locus with exactly one gliding edge has dimension \(8\), and the
locus with two rank-one gliding edges has dimension \(7\).

Finally, the SWAP locus is characterized by \(A=D=0\).  The remaining
blocks \(B,C\in\SL_2(K)\) are independent, so this locus has dimension
\[
3+3=6.
\]
The conditions
\[
 \det A\ne0,\qquad
 \det B\ne0,\qquad
 \det C\ne0,\qquad
 \det D\ne0
\]
define principal Zariski-open subsets of \(\Sp_4(K)\).  Their
intersection is nonempty by the perfect-core parameterization, so
the perfect locus is a nonempty Zariski-open subset.

It remains to prove density.  The group \(\Sp_4(K)\) is generated by
symplectic transvections
\cite[Theorem~8.5, p.~72]{taylor1992geometry}.
For each fixed \(u\ne0\), the transformations
\[
 T_{u,a}(v)=v+a\langle v,u\rangle u,
 \qquad a\in K,
\]
form a closed one-parameter subgroup
\[
 U_u=\{T_{u,a}:a\in K\}
\]
isomorphic to the additive algebraic group \(\mathbb G_a\).
The group \(\mathbb G_a\) is irreducible and one-dimensional
\cite[Sec.~7.1, p.~51]{humphreys2012linear}, and hence \(U_u\) is
connected.  Since the subgroups \(U_u\) generate \(\Sp_4(K)\),
Humphreys's connected-generation criterion implies that
\(\Sp_4(K)\) is connected
\cite[Sec.~7.5, Corollary, p.~56]{humphreys2012linear}.

For an algebraic group, the identity component is the unique
irreducible component containing the identity, and its cosets are
both the irreducible and connected components
\cite[Sec.~7.3, p.~53]{humphreys2012linear}.
The connected algebraic group \(\Sp_4(K)\) is
irreducible.  Every nonempty open subset of an irreducible variety
is dense
\cite[Sec.~1.3, p.~7]{humphreys2012linear}.
The perfect locus is therefore Zariski dense in \(\Sp_4(K)\).

This is a geometric statement over \(K\).  The open locus has no
\(\F_2\)-rational points, consistently with the absence of the
qubit perfect phase.
\end{proof}

\begin{corollary}
Let
\[
Q=q^2
\]
denote the total Hilbert-space dimension of the two-qudit system, and
let \(\overline{\mathrm{Cl}}_{\mathrm{FF}}(2,q)\) denote the preimage
of the extension-field subgroup \(\Sp(4,q)\) under the projective
Clifford-to-symplectic quotient map.  The projective two-qudit Weyl
group is labeled by the additive phase space
\[
(\F_q^4,+),
\]
equivalently by \(\F_p^{4m}\) when \(q=p^m\).  It therefore has order
\[
|\F_q^4|=q^4=Q^2.
\]
The restricted projective Clifford group consequently fits into the
exact sequence
\[
1\longrightarrow (\F_q^4,+)
\longrightarrow \overline{\mathrm{Cl}}_{\mathrm{FF}}(2,q)
\longrightarrow \Sp(4,q)
\longrightarrow 1
\]
\cite[Sec.~II, the two paragraphs following Eq.~(7), and the
paragraph following Eq.~(10)]{Zhu2017Multiqubit}.

Thus every \(S\in\Sp(4,q)\) has exactly \(q^4\) projective Clifford
lifts, obtained by composing any fixed lift of \(S\) with the
\(q^4\) projective Weyl displacements.  This factor \(q^4\) is the
cardinality of the Weyl kernel and is independent of the factor
\(q^4\) appearing in the order formula for \(\Sp(4,q)\).  Since Weyl
displacements do not change the transport-phase label, all
probabilities in \Cref{thm:concentration,cor:rarity} hold unchanged
for a uniformly random finite-field projective Clifford gate.

For this subgroup, acting on the \(Q=q^2\)-dimensional two-qudit
Hilbert space, Zhu's frame-potential formulas give
\[
\Phi_2=2,
\qquad
\Phi_3=2q+2.
\]
Thus this ensemble is an exact unitary \(2\)-design for every prime
power \(q\), but a unitary \(3\)-design only for \(q=2\)
\cite[Theorem~1, Lemma~1, Eqs.~(8)--(11), and the paragraph following
Eq.~(11)]{Zhu2017Multiqubit}.  Therefore the large-\(q\) concentration
of the perfect-tensor phase is an ensemble-specific algebraic result,
not a consequence of third-moment Haar matching.
\end{corollary}

\subsection{Sampling perfect tensors}

\Cref{cor:compilation} gives a local-orbit sampler that is especially convenient for compilation.  The counting proof gives a second direct sampler in determinant coordinates.

\begin{corollary}
\label{cor:sampler}
The following procedure samples uniformly from the perfect-tensor symplectic classes.
\begin{enumerate}[label=\arabic*.]
\item Choose $\delta$ uniformly from $\F_q\setminus\{0,1\}$.
\item Choose $A,B,C$ independently and uniformly subject to
\[
 \det A=1-\delta,
 \qquad
 \det B=\det C=\delta.
\]
\item Set $D$ by \eqref{eq:Dformula}.
\end{enumerate}
Every perfect-tensor symplectic class is generated exactly once.
\end{corollary}

A matrix with prescribed nonzero determinant $a$ can itself be sampled without rejection.  Choose a nonzero first column $u\in\F_q^2$, choose one $v_0$ with $\det(u,v_0)=a$, choose $c\in\F_q$ uniformly, and set $v=v_0+cu$.  The matrix $(u\ v)$ is uniform in the determinant-$a$ fiber because each nonzero $u$ has exactly $q$ admissible second columns.

\section{Edge erasure and causal local masking}
\label{sec:masking}

The edge channels act on one-site Weyl operators that remain confined to a light-cone edge.  Immediate erasure means no nonidentity one-site Weyl survives one edge step.  Delayed erasure means a mode may survive once but not twice.

\subsection{Perfect-tensor causal saturation}

\begin{lemma}
\label{lem:perfectpeeling}
Consider a brickwork circuit whose two-site gate is a Clifford perfect tensor.  In the doubled tensor network for a reduced output channel, every gate with two contracted legs can be canceled in any of the three possible $2|2$ orientations.  Starting from one output site, the greedy backward peeling region reaches exactly $4t$ consecutive input sites after $t$ periods.  Starting from two output sites, the smallest possible nontrivial peeling region has length $4t-2$.
\end{lemma}

\begin{proof}
A perfect tensor is unitary across every bipartition of two legs against two legs.  In the doubled network, contracting a tensor with its conjugate across any two corresponding legs cancels the pair.  For one output site, the first backward half-layer exposes the two legs of one gate, after which each alternating half-layer exposes one new gate at each boundary.  The peeled input interval has lengths
\[
 2,4,6,\ldots,4t
\]
after the successive half-layers.

For two output sites, the smallest region is obtained when the two sites occupy one gate in the first backward layer and are chosen as the two-leg image of a one-site Weyl.  The first gate then peels to one effective site.  The remaining $2t-1$ half-layers produce an interval of length $4t-2$.  Any other placement exposes at least as many boundary legs.  Write the
inverse symplectic action of one gate as
\[
S^{-1}
=
\begin{pmatrix}
E&F\\
G&H
\end{pmatrix}.
\]
Because the gate is a perfect tensor, all four blocks \(E,F,G,H\) are
invertible.  At an exposed boundary gate, the incoming two-site label
has the form \((x,0)\) or \((0,x)\) with \(x\neq0\).  Its backward image
is therefore, respectively,
\[
(Ex,Gx)
\qquad\text{or}\qquad
(Fx,Hx).
\]
Both components are nonzero.  In particular, the newly created
exterior label is nonzero.  That exterior site is touched by only this
boundary gate during the active half-layer, so no label coming from the
interior can reach the same site and cancel it.  Hence every exposed
boundary advances by one site at each subsequent half-layer, and the
backward support cannot terminate before reaching the stated boundary.
Therefore no smaller nontrivial backward support is possible.
\end{proof}

\begin{theorem}
\label{thm:perfectmasking}
For every finite-field Clifford perfect tensor and every $t\ge1$,
\[
 \boxed{d_1(t)=4t,\qquad d_2(t)=4t-2.}
\]
Consequently a payload of length $4t-1$ is perfectly hidden from every one-site observer, and a payload of length $4t-3$ is perfectly hidden from every two-site observer.  These values saturate the causal ceilings \eqref{eq:causalceilings}.
\end{theorem}

\begin{proof}
The peeling regions in \Cref{lem:perfectpeeling} give the lower bounds.  Equality for $d_1$ is realized by the backward image of any one-site Weyl operator.  Equality for $d_2$ is realized by choosing a two-site output Weyl that is the image, under the final gate, of a one-site Weyl and then pushing it through the remaining layers.
\end{proof}

For qubits, the perfect-tensor phase is impossible because four-party perfect tensors are equivalent to \(\operatorname{AME}(4,q)\) states, while an \(\operatorname{AME}(4,2)\) state does not exist \cite[p.~2]{higuchi2000entangled}.  A recent direct account of this nonexistence result is provided in
\cite[Observation~1 and Proofs~1--7]{Huber_2026}.
For qubits, the first available finite erasure depth is therefore \(\kappa=2\).

\subsection{Half-layer support growth}

Let
\[
 S^{-1}=R=\begin{pmatrix}E&F\\G&H\end{pmatrix}.
\]
Dual-unitarity makes the cross blocks $F$ and $G$ invertible.  A finite support interval is called outward-facing for a half-layer when its left endpoint is the right leg of an active bond and its right endpoint is the left leg of an active bond.

\begin{lemma}
\label{lem:survivalchains}
Fix either light-cone edge.  A nonidentity one-site Weyl label remains one-site through $s$ consecutive inverse half-layers along that edge if and only if the corresponding edge transfer map has a nonzero $s$th iterate on that label.  Thus,
\begin{enumerate}[label=(\alph*)]
\item in the perfect phase, no one-site label survives even one half-layer;
\item in the delayed phase, some label survives one half-layer but no label survives two;
\item in a glider phase, a nonzero projective line survives indefinitely.
\end{enumerate}
\end{lemma}

\begin{proof}
On one active bond, backward conjugation of a one-site label is obtained by applying one column pair of $R$.  It remains one-site exactly when the component on the inward leg vanishes.  The surviving exterior label is precisely the partial Weyl map defining the appropriate edge channel.  At the next staggered layer the same statement is applied to that surviving label.  Iterating identifies an $s$-step one-site path with a nonzero value of $T_\nu^s$.  The three conclusions follow from $T_\nu=0$, from $T_\nu\ne0$ but $T_\nu^2=0$, or from the invariant projective line of a nonnilpotent rank-one partial permutation.
\end{proof}

\begin{lemma}
\label{lem:nocollision}
In the delayed phase, two isolated one-site labels that survive inward from adjacent active bonds cannot collide into a one-site label at the next half-layer.
\end{lemma}

\begin{proof}
Support is unchanged by invertible one-site changes of Weyl coordinates, so normalize the rank-one block to
\[
 A=e_1e_1^T,
 \qquad
 B=\begin{pmatrix}a&b\\c&d\end{pmatrix},
 \qquad
 C=\begin{pmatrix}e&f\\g&h\end{pmatrix},
 \qquad
 B,C\in\SL(2,q).
\]
The delayed conditions say that $B$ does not fix $[e_1]$ and $C$ does not fix $[Je_1]=[e_2]$, hence
\[
 c\ne0,
 \qquad
 f\ne0.
\]
The inverse-gate blocks are
\[
 E=A^\#=\begin{pmatrix}0&0\\0&1\end{pmatrix},
 \qquad F=C^{-1},
 \qquad G=B^{-1},
\]
and, from $D=-Ce_2e_2^TB$, the second column of $H=D^\#$ is
\[
 He_2=fGe_1.
\]
The two inward-surviving labels become, up to nonzero amplitudes, the adjacent pair
\[
 (\alpha Ge_1,\,\beta e_2).
\]
After the next inverse gate its left and right labels are
\[
 \ell=E(\alpha Ge_1)+F(\beta e_2)
 =\begin{pmatrix}-\beta f\\-\alpha c+\beta e\end{pmatrix},
\]
\[
 r=G(\alpha Ge_1)+H(\beta e_2)
 =G(\alpha Ge_1+\beta f e_1).
\]
If $\ell=0$, then $f\ne0$ gives $\beta=0$, and $c\ne0$ then gives $\alpha=0$.  If $r=0$, invertibility of $G$ and linear independence of $Ge_1$ and $e_1$ again give $\alpha=\beta=0$.  Thus no nonzero incoming pair can lose either output label.
\end{proof}

\begin{lemma}
\label{lem:halflayer}
Let $m_r(n)$ be the minimum span after $n$ alternating inverse half-layers, starting from a nonidentity Weyl string of weight at most $r$.

If $C$ is a perfect tensor, then
\[
 m_1(n)=2n\quad(n\ge1),
\]
\[
 m_2(1)=1,\qquad m_2(n)=2n-2\quad(n\ge2).
\]

If $C$ is a delayed bilateral eraser, then
\[
 m_1(1)=1,\qquad m_1(n)=2n-2\quad(n\ge2),
\]
\[
 m_2(1)=m_2(2)=1,\qquad m_2(n)=2n-4\quad(n\ge3).
\]
\end{lemma}

\begin{proof}

Whenever an interval is outward-facing, each exterior output site is
touched by exactly one endpoint label in the active half-layer.  No
interior label contributes to either exterior site, and the invertible
cross blocks $F$ and $G$ map the two nonzero endpoint labels to nonzero
exterior labels.  The next interval is outward-facing for the alternating
layer, so every later half-layer increases the span by exactly two.

For a perfect tensor, $E$ and $H$ are also invertible.  A one-site label therefore becomes a two-site outward-facing interval after one half-layer, proving $m_1(n)=2n$.  A weight-two operator placed on one active bond can be chosen as the inverse image of a one-site label, so it compresses to one site in the first half-layer and then follows the one-site profile.  Any operator occupying two different active bonds produces nonzero exterior labels through $F$ and $G$ and is strictly larger.  This proves the perfect profiles.

For a delayed eraser, $E$ and $H$ have rank one.  By \Cref{lem:survivalchains}, a one-site label can survive exactly one inverse half-layer but cannot survive the next.  When the survival chain terminates, both outputs of the active gate are nonzero; they form a two-site outward-facing interval.  Every later half-layer therefore expands the span by two.  This gives the $m_1$ profile.

A weight-two string placed on one active bond may compress to one site.  It can then follow one exceptional one-site edge step, but a second such survival is forbidden by \Cref{lem:survivalchains}.  Hence by the third half-layer it has become a two-site outward-facing interval.  Conversely, choosing the inverse image of the exceptional edge line realizes two consecutive one-site stages.

It remains to exclude a better history from two different active bonds.  If an extreme label does not survive inward, that side creates an outward-facing front immediately.  If both extremes survive inward and the bonds are separated, both survivals terminate at the next layer and the resulting exterior fronts are farther apart than the claimed minimum.  If the bonds are adjacent, the two surviving labels meet on one active bond, but \Cref{lem:nocollision} shows that they cannot compress there.  The next interval is therefore outward-facing and is again no smaller than the same-bond construction.  This proves the $m_2$ profile.
\end{proof}

\begin{theorem}
\label{thm:deficit}
Let $C$ be a finite-field Clifford dual-unitary gate over $\F_q$.
\begin{enumerate}[label=(\roman*)]
\item If $\kappa(C)=1$, then for every $t\ge1$,
\[
 d_1(t)=4t,\qquad d_2(t)=4t-2.
\]
\item If $\kappa(C)=2$, then
\[
 d_1(t)=4t-2\quad(t\ge1),
\]
\[
 d_2(1)=1,\qquad d_2(t)=4t-4\quad(t\ge2).
\]
\item If $\kappa(C)=\infty$, then a one-site Weyl glider exists and
\[
 d_r(t)=1
\]
for every $r,t\ge1$.
\end{enumerate}
In the finite-erasure phases and away from the single depth-one exception,
\[
 \boxed{d_r(t)=4t-2(r-1)-2(\kappa-1),\qquad r\in\{1,2\}.}
\]
Thus edge-erasure depth is exactly one plus half the deficit from the causal masking ceiling.
\end{theorem}

\begin{proof}
One Floquet period consists of two half-layers, so substitute $n=2t$ into \Cref{lem:halflayer}.  In the nonnilpotent phases, \Cref{lem:edgemaps} supplies a projective line on which the edge partial permutation closes.  The associated one-site Weyl operator translates indefinitely without spreading.
\end{proof}

\section{Local-observer private encoding}
\label{sec:privacy}

The masking distances are not only operator-spreading diagnostics.  By \Cref{thm:referenceprivacy}, they are exact privacy thresholds for the reversible encoder \eqref{eq:privateencoder}.  Define the $r$-local privacy radius
\[
 p_r^C(t)=d_r^C(t)-1,
\]
the largest integer message length that is guaranteed to be invisible to every output set of at most $r$ qudits.

\begin{corollary}
\label{cor:privacyradii}
For every prime power $q$, the finite-erasure phases have
\[
\begin{array}{c|cc}
\text{transport phase}
& p_1(t)&p_2(t)\\
\hline
\kappa=1
&4t-1&4t-3\\
\kappa=2,\ t\ge2
&4t-3&4t-5.
\end{array}
\]
At depth $t=1$, a delayed eraser has $p_1(1)=1$ and $p_2(1)=0$.  Every gliding phase has $p_r(t)=0$ for all $r,t\ge1$.

If a message interval $M$ obeys the relevant bound and $A$ is any output set with $|A|\le r$, then
\[
 (\id_R\otimes\cN_{A\leftarrow M}^{(t)})(\rho_{RM})
 =\rho_R\otimes I_A/q^{|A|}
\]
for every reference-entangled input $\rho_{RM}$.  In the perfect phase these radii are optimal within the depth-$2t$ nearest-neighbor brickwork architecture because they saturate the causal ceilings \eqref{eq:causalceilings}.
\end{corollary}

\begin{proof}
The formulas are $d_r(t)-1$ using \Cref{thm:deficit}.  Exact decoupling follows from \Cref{thm:referenceprivacy}.  The causal optimality statement follows because no circuit in the architecture can have $d_1(t)>4t$ or $d_2(t)>4t-2$.
\end{proof}

\subsection{Operational interpretation and limitations}

The encoder supplies a mixed-state spatial privacy structure rather
than a complete quantum secret-sharing access structure.  In the
terminology of Cleve, Gottesman, and Lo, an authorized set can
reconstruct the secret, whereas an unauthorized set has a reduced
state that is independent of the secret.  A \(((k,n))\) threshold
scheme requires every coalition of at least \(k\) shares to be
authorized and every coalition of at most \(k-1\) shares to be
unauthorized
\cite[definition of a threshold scheme and discussion of general
access structures]{cleve1999share}.

For comparison, a \(2n\)-party perfect tensor, regarded as an isometry
from one logical qudit to \(2n-1\) physical qudits, defines a
\([[2n-1,1,n]]_q\) erasure-correcting code and a
\(((n,2n-1))\) threshold scheme: any \(n-1\) shares contain no
information about the logical qudit, while any \(n\) shares can
reconstruct it
\cite[Sec.~2]{pastawski2015holographic}.  Here,
\Cref{thm:referenceprivacy} proves only that every output set of size
at most \(r\) is unauthorized below the stated radius, while
\Cref{cor:fulldecoder} proves that the full output set is authorized.

The independent maximally mixed exterior supplies the randomness, so
the construction differs from the pure-ancilla deterministic masking
setting studied in \cite{modi2018masking,li2018multipartite,shi2021k}.  It is also
operationally distinct from standard bipartite quantum data hiding.
In the latter setting, the adversaries collectively possess the full
bipartite state but are restricted to LOCC measurements
\cite[Secs.~I--III]{divincenzoquantum}.  Here, an observer has
unrestricted access to a bounded output set \(A\), but no access to
\(A^c\).  The guarantee
\[
(\id_R\otimes\cN_{A\leftarrow M}^{(t)})(\rho_{RM})
=
\rho_R\otimes\frac{I_A}{q^{|A|}}
\]
therefore implies perfect indistinguishability under every measurement
on \(A\), including when the encoded input is entangled with a
reference.

The closest masking analogue is \(k\)-uniform quantum information
masking, where every reduction to any chosen \(k\) parties is
independent of the encoded state
\cite[Definition~9]{shi2021k}.  Here,
\Cref{thm:referenceprivacy} establishes this exact decoupling for every
output set of size at most \(r\), while coalitions of larger
intermediate size are not classified.  We therefore use the term
\emph{local-observer private encoding} rather than either a standard
LOCC data-hiding protocol or a threshold quantum secret-sharing
scheme.

The result establishes exact low-locality decoupling, not Haar-like
scrambling or privacy against arbitrary large coalitions.  Quantum convolutional codes offer a useful comparison because they also encode a spatially ordered stream by repeatedly applying
finite-range Clifford operations.  In an early concrete stabilizer
construction, Ollivier and Tillich described a rate-\(1/5\) quantum
convolutional code whose stabilizer generators and logical Pauli
operators are periodic translates of finite-support operators.  Their
construction admits online encoding and, for the exhibited
non-catastrophic encoder, forward online decoding.  They also gave a
maximum-likelihood error-estimation algorithm whose classical
complexity is linear in the number of encoded qubits for memoryless
channels
\cite[pp.~1--4,  Eqs.~(2) and~(7), and the sections
``Error propagation and online decoding'' and
``Maximum likelihood error estimation'']{OllivierTillich2003}.

The masking distances \(d_r(t)\),
however, are not convolutional-code distances.  They measure the
smallest input span of the backward image of an output Weyl observable
of weight at most \(r\), and therefore determine an exact
local-decoupling threshold.  Whether the higher-local quantities
\(d_r(t)\) admit a coding-theoretic interpretation analogous to a
quantum column-distance profile remains open.  Stabilizer subsystem
codes provide a second comparison target.  In the subsystem stabilizer
formalism, the code space decomposes as
\[
\mathcal C=\mathcal A\otimes\mathcal B,
\]
where the logical information is stored in \(\mathcal A\), while gauge
transformations act trivially on \(\mathcal A\) and only on the gauge
subsystem \(\mathcal B\). The logical Pauli group is represented by the
quotient \(N(\mathcal S)/\mathcal G\), and the subsystem-code distance is
the minimum Pauli weight in \(N(\mathcal S)\setminus\mathcal G\)
\cite[pp.~2--3, Eqs.~(2)--(3), and the paragraph titled
``Bounds'']{poulin2005stabilizer}. This gives a concrete code-theoretic
comparison target, although the masking distances \(d_r(t)\) defined
here are not identified with subsystem-code distances. In classical
convolutional coding, a maximum-distance-profile code has column
distances that attain their largest permitted values throughout the
maximal initial window
\cite[Sec.~II, especially Definition~2.9 and
Theorem~2.12]{gluesing2006strongly}. Quantum burst-error-correcting codes provide another spatially local comparison. They correct Pauli errors whose nonidentity support is confined to a contiguous interval and satisfy the quantum Reiger bound
\[
n-k\ge 4\ell
\]
when every burst of length at most \(\ell\) is correctable
\cite[Secs. II.A--B and Theorem 3]{fan2018construction}.
Unlike a burst-error-correction distance, \(d_r(t)\) measures the
smallest input span of the backward image of a low-weight output
observable and therefore determines a local-decoupling threshold.  The
comparisons above are structural only; no direct coding-theoretic
equivalence is claimed.

\section{Qutrit implementation and error sensitivity}
\label{sec:qutritdemo}

The smallest local dimension supporting the immediate-erasure phase is
$q=3$.  In this case there is only one ordered perfect core, since
$\F_3\setminus\{0,1\}=\{2\}$.  The canonical matrix
\eqref{eq:canonicalcore} becomes
\begin{equation}
S_3=
\begin{pmatrix}
2&0&2&0\\
0&1&0&1\\
2&0&1&0\\
0&1&0&2
\end{pmatrix}
\in\Sp(4,3),
\label{eq:qutritS3}
\end{equation}
where the Weyl coordinates are ordered as
$(x_1,z_1,x_2,z_2)$.  This representative has the following two-gate unitary realization.  Let
\begin{equation}
M=
\begin{pmatrix}
2&2\\
2&1
\end{pmatrix}
\in\SL(2,3),
\qquad
U_3\lvert a,b\rangle
=
\lvert 2a+2b,\,2a+b\rangle,
\label{eq:qutritU3}
\end{equation}
with all arithmetic modulo three.  Since
\[
M^{-T}=
\begin{pmatrix}
1&1\\
1&2
\end{pmatrix},
\]
the action of $U_3$ on the $x$ and $z$ Weyl exponents is exactly
\eqref{eq:qutritS3}.  Equivalently,
\begin{equation}
U_3
=
\operatorname{SUM}_{2\to1}^{-1}
\operatorname{SUM}_{1\to2}^{-1},
\label{eq:twoSUMimplementation}
\end{equation}
where
\[
\operatorname{SUM}_{1\to2}^{-1}\lvert a,b\rangle
=\lvert a,b-a\rangle,
\qquad
\operatorname{SUM}_{2\to1}^{-1}\lvert a,b\rangle
=\lvert a-b,b\rangle.
\]
Thus one perfect-tensor brick requires only two inverse qutrit-SUM
gates.  This realization is locally Clifford-equivalent to the
two-controlled-adder construction of the four-qutrit perfect tensor
given in \cite[App.~A.2, Eq.~(A.11)]{pastawski2015holographic}.  Indeed,
over \(\F_3\),
\[
\begin{pmatrix}
2&2\\
2&1
\end{pmatrix}
=
\begin{pmatrix}
1&0\\
0&2
\end{pmatrix}
\begin{pmatrix}
2&1\\
1&1
\end{pmatrix}
\begin{pmatrix}
1&0\\
0&2
\end{pmatrix},
\]
where the middle matrix is the qutrit perfect-tensor map of
\cite[Eq.~(A.11)]{pastawski2015holographic}, and the two diagonal
matrices correspond to one-qutrit inversions on the second input and
second output. A complementary representative of the same qutrit perfect core arises from the quadratic-phase Bell-basis construction of Yu, Wang, and Kos
\cite[Sec.~2.2, Eqs.~(11), (19), and (20)]{Yu2024hierarchical}.
By \Cref{thm:fullcores}, their \(D=3\) example is ordered locally
Clifford-equivalent to \(U_3\).

The normalized Choi state of \(U_3\) is the nine-term stabilizer state
\begin{equation}
\lvert U_3\rangle
=
\frac13\sum_{a,b\in\F_3}
\lvert a,b,2a+2b,2a+b\rangle,
\label{eq:qutritAMEstate}
\end{equation}
which is an \(\operatorname{AME}(4,3)\) state.  Directly reshaping the four-index tensor in the three inequivalent \(2|2\) orientations gives three \(9\times 9\) matrices proportional to unitaries; multiplying each reshape by \(3\) gives a unitary matrix.

\subsection{Exact support and channel verification}

We implemented the inverse brickwork dynamics directly on Weyl labels.
For every depth $t\in\{1,2,3\}$, the search enumerated all nonidentity
one-site labels and all pairs of nonidentity two-site labels on both
output sublattices.  The computed minima are
\begin{equation}
\begin{array}{c|cc|cc}
&\multicolumn{2}{c|}{\text{qutrit perfect core}}
&\multicolumn{2}{c}{\text{qubit delayed eraser}}\\
 t&d_1(t)&d_2(t)&d_1(t)&d_2(t)\\
\hline
1&4&2&2&1\\
2&8&6&6&4\\
3&12&10&10&8
\end{array}
\label{eq:qutritdistancecheck}
\end{equation}
The qutrit values reproduce the causal optima
$d_1(t)=4t$ and $d_2(t)=4t-2$, while the binary representative
reproduces the delayed-erasure deficits.  A separate exhaustive
enumeration of the determinant parameterization gives the five
$q=3$ phase counts
\begin{equation}
13824,
\quad10368,
\quad6912,
\quad1152,
\quad576,
\label{eq:qutritphasecheck}
\end{equation}
in the order immediate erasure, delayed erasure, one gliding edge,
two gliding edges, and SWAP transport.

We also evaluated the full channel matrix for a four-qutrit periodic
brickwork circuit after one period.  For a channel
$\cN_{A\leftarrow M}$, define the normalized Choi leakage
\begin{equation}
\ell(A\leftarrow M)
=
\frac12\left\|
J(\cN_{A\leftarrow M})
-
\frac{I_R}{3^{|M|}}\otimes\frac{I_A}{3^{|A|}}
\right\|_1.
\label{eq:choileakage}
\end{equation}
This quantity vanishes exactly when the observed channel is completely
depolarizing, including for inputs entangled with the reference $R$.
For a contiguous three-qutrit message, the maximum over all four
one-site observers was
\[
\max_{|A|=1}\ell(A\leftarrow M)
=7.1\times10^{-17}.
\]
For a one-qutrit message, the maximum over all six two-site sets,
including disconnected pairs, was
\[
\max_{|A|=2}\ell(A\leftarrow M)
=1.9\times10^{-16}.
\]
Applying the exact inverse circuit produced the identity recovery
channel, and hence both its entanglement fidelity and average recovery
fidelity are equal to one
\cite[Eqs.~(2)--(3)]{NIELSEN2002249}.
The residual leakage is therefore numerical roundoff rather than
physical leakage.
\subsection{Coherent calibration errors}

To test sensitivity away from the perfect-tensor manifold, every
brick was replaced by
\begin{equation}
\widetilde U_3(\epsilon)
=
\exp(-i\epsilon H)U_3,
\qquad
H=\lvert00\rangle\!\langle11\rvert
 +\lvert11\rangle\!\langle00\rvert.
\label{eq:coherenterrormodel}
\end{equation}
This perturbation mixes two computational-basis amplitudes and is not
a phased permutation.  For generic nonzero \(\epsilon\), it breaks both nontrivial reshuffled-unitarity conditions.  The forward circuit used
$\widetilde U_3(\epsilon)$, while the authorized receiver applied the
ideal inverse.  Let $\mathcal R_\epsilon$ denote the induced decoded
channel on the one-qutrit message.  We measure its recovery quality by
the standard average gate fidelity relative to the identity channel,
\begin{equation}
F_{\rm avg}(\mathcal R_\epsilon)
=
\int d\psi\,
\left\langle\psi\middle|
\mathcal R_\epsilon
\bigl(\lvert\psi\rangle\!\langle\psi\rvert\bigr)
\middle|\psi\right\rangle,
\label{eq:averagerecoveryfidelity}
\end{equation}
where $d\psi$ is normalized Haar measure
\cite[Eqs.~(1)--(2)]{NIELSEN2002249}.
The one-site leakage uses a maximal three-qutrit payload, while the
two-site leakage and $F_{\rm avg}$ use a one-qutrit payload.

\begin{table}[t]
\centering
\begin{tabular}{cccc}
\toprule
$\epsilon$
&$\max_{|A|=1}\ell(A\leftarrow M_3)$
&$\max_{|A|=2}\ell(A\leftarrow M_1)$
&$F_{\rm avg}$\\
\midrule
0     &$7.0\times10^{-17}$&$1.9\times10^{-16}$&1.000000\\
0.01  &$3.878\times10^{-3}$&$2.222\times10^{-3}$&0.999956\\
0.02  &$7.756\times10^{-3}$&$4.444\times10^{-3}$&0.999822\\
0.05  &$1.938\times10^{-2}$&$1.111\times10^{-2}$&0.998890\\
0.10  &$3.869\times10^{-2}$&$2.219\times10^{-2}$&0.995574\\
0.20  &$7.685\times10^{-2}$&$4.415\times10^{-2}$&0.982518\\
\bottomrule
\end{tabular}
\caption{Reference-sensitive local leakage and ideal-inverse recovery
for the coherent error model \eqref{eq:coherenterrormodel}.}
\label{tab:qutritcoherenterror}
\end{table}

For $\epsilon\le0.05$, least-squares fits through the origin give
\begin{equation}
\max_{|A|=1}\ell(A\leftarrow M_3)
\simeq0.388\epsilon,
\qquad
\max_{|A|=2}\ell(A\leftarrow M_1)
\simeq0.222\epsilon,
\label{eq:leakagefit}
\end{equation}
and
\begin{equation}
1-F_{\rm avg}
\simeq0.444\epsilon^2.
\label{eq:recoveryfit}
\end{equation}
Thus coherent miscalibration is first visible in local privacy, while
the ideal-inverse recovery infidelity begins quadratically.  

\subsection{Stochastic Pauli noise}

A different behavior occurs for stochastic Weyl errors.  After each
half-layer, independently on every qutrit, apply the identity with
probability $1-p$ and one of the eight nonidentity qutrit Weyl
operators uniformly with total probability $p$.  Every fixed error
pattern can be propagated to the final time through the intervening
Clifford gates, so the noisy output has the form
\begin{equation}
P_{\boldsymbol e}\,
\cU_{U_3}^{t}(\rho)\,
P_{\boldsymbol e}^{\dagger}
\label{eq:paulipropagation}
\end{equation}
for a message-independent tensor-product Weyl operator
$P_{\boldsymbol e}$.  Therefore every output set that was exactly
maximally mixed in the ideal circuit remains exactly maximally mixed
for every error pattern.  Stochastic Pauli noise degrades authorized
recovery but does not create local information leakage.

The recovered one-qutrit channel is a Weyl--Pauli channel of the form
\begin{equation}
\mathcal P_p(\rho)
=
\sum_{a\in\F_3^2}
\pi_a W_a\rho W_a^\dagger,
\qquad
\sum_{a\in\F_3^2}\pi_a=1.
\label{eq:recoveredpaulichannel}
\end{equation}
Let
\[
\lvert\Phi_3\rangle
=
\frac{1}{\sqrt3}
\sum_{j=0}^{2}\lvert j,j\rangle.
\]
Since every nonidentity qutrit Weyl operator is traceless,
\begin{align}
F_e(\mathcal P_p)
&=
\sum_{a\in\F_3^2}
\pi_a
\left|
\langle\Phi_3|
(I\otimes W_a)
|\Phi_3\rangle
\right|^2
\nonumber\\
&=
\sum_{a\in\F_3^2}
\pi_a
\left|
\frac{\operatorname{tr}(W_a)}{3}
\right|^2
=
\pi_0.
\label{eq:paulientanglementfidelity}
\end{align}
Thus the entanglement fidelity is exactly the probability that the
propagated Weyl label is trivial on the message site.  The general
qudit relation between average gate fidelity and entanglement fidelity
then gives
\begin{equation}
F_{\rm avg}(\mathcal P_p)
=
\frac{3F_e(\mathcal P_p)+1}{4}
=
\frac{3\pi_0+1}{4}
\label{eq:qutritaveragefidelity}
\end{equation}
(see \cite[Eq.~(3)]{NIELSEN2002249}).
Exact convolution over all $3^8$ four-qutrit Weyl labels gives
\begin{equation}
\pi_0^{(t)}(p)
=
\frac{1+8\left(1-\frac{9p}{8}\right)^{n_t}}{9},
\qquad
(n_1,n_2,n_3)=(6,9,15).
\label{eq:qutritidentityprobability}
\end{equation}
Substitution into \eqref{eq:qutritaveragefidelity} gives
\begin{table}[h]
\centering
\begin{tabular}{c|cccc}
\toprule
&\multicolumn{4}{c}{$F_{\rm avg}$}\\
$t\backslash p$&0.001&0.01&0.02&0.05\\
\midrule
1&0.995513&0.956247&0.914913&0.804365\\
2&0.993280&0.935459&0.876533&0.729266\\
3&0.988838&0.895942&0.807204&0.613079\\
\bottomrule
\end{tabular}
\caption{One-qutrit recovery after independent single-site Pauli noise
of total probability $p$ following every half-layer.  The local Choi
leakage remains exactly zero.}
\label{tab:qutritpaulinoise}
\end{table}

The two error models affect the protocol differently.  Perfect-tensor
geometry gives exact local privacy and saturates the nearest-neighbor
causal ceiling.  Coherent perturbations transverse to the perfect
manifold create local leakage at first order, whereas stochastic Pauli
errors preserve local privacy but reduce the fidelity of the inverse
decoder.  Comparing local leakage with recovery fidelity therefore
distinguishes the two error models considered here.

\section{The qubit case}
\label{sec:qubit}

For $q=2$, the edge matrices act on the three Pauli axes.  Define
\[
 \Gamma(C)=\Frob{T_L(C)^2}^2+\Frob{T_R(C)^2}^2.
\]
The exact class counts are
\[
 \Gamma\in\{0,1,2,6\},
 \qquad
 \#\{\Gamma=0,1,2,6\}=144,144,36,36.
\]

\begin{proposition}
\label{prop:glider}
If $\Gamma(C)>0$, then the brickwork circuit contains a one-site Pauli glider.  Hence
\[
 d_r^C(t)=1
\]
for every $r\ge1$ and $t\ge1$.
\end{proposition}

\begin{proof}
In an iSWAP-core edge channel, $T=\pm e_ae_b^T$.  The condition $T^2\ne0$ is $a=b$, in which case that Pauli axis is translated along the edge without spreading.  In the SWAP sector all three axes are transported without spreading.
\end{proof}

The \(144\) classes with \(\Gamma=0\) reduce under local Pauli-frame changes, translation, and reflection to three representatives.  These representatives correspond to the three iSWAP-core classes termed ``good scramblers'' by Sommers, Huse, and Gullans \cite[Secs.~VI.B and VI.D]{Sommers2023Crystalline}.  Our classification also retains the exact symplectic-class count and
the parity-tagged boundary dynamics.  The representative binary symplectic matrices and boundary automata are listed in \cref{app:boundary}.

\begin{theorem}
\label{thm:qubitmain}
Let $C$ be a two-qubit Clifford dual-unitary gate.  For every $t\ge2$,
\[
 \boxed{
 \Gamma(C)=0
 \iff
 d_1^C(t)=4t-2,
 \quad
 d_2^C(t)=4t-4.
 }
\]
If $\Gamma(C)>0$, then $d_r^C(t)=1$ for all $r,t\ge1$.
Consequently
\[
 \max_{C\ \mathrm{Clifford\ DU}}d_2^C(t)=4t-4,
\]
and the maximizers are exactly the $144$ bilaterally nilpotent symplectic classes.
\end{theorem}

\begin{proof}
For $q=2$ the perfect phase is empty, the classes with $\Gamma=0$ are exactly the delayed bilateral erasers, and every nonzero value of $\Gamma$ is a gliding phase.  The claim is therefore the $q=2$ specialization of \Cref{thm:deficit}.
\end{proof}

\begin{proposition}
\label{prop:threshold}
Let $C$ be any bilaterally nilpotent two-qubit Clifford dual-unitary gate and let $t\ge2$.  Fix a threshold input interval $M$ of length $4t-4$ with the parity alignment attaining the minimum.  Up to phases, exactly two weight-two output Pauli observables, denoted $G_t^{(1)}$ and $G_t^{(2)}$, have backward images supported inside $M$.  They commute and generate
\[
 \operatorname{alg}\bigl(G_t^{(1)},G_t^{(2)}\bigr)\cong\mathbb C^4.
\]
Their product has output weight four.  Thus the set of leaked weight-at-most-two observables is not itself a vector space.  Rather, the collection of two-local probes first exposes two commuting binary parity generators.  A single fixed two-site observer need not access both generators, and the opposite interval parity need not attain the same threshold structure.
\end{proposition}

\begin{proof}
The \(144\) bilaterally nilpotent classes reduce, under local
Pauli-frame changes, translation, and reflection, to the three
representatives
\[
S_{00},\qquad S_{03},\qquad S_{33}
\]
listed in \Cref{app:boundary}.  It is therefore enough to perform the
depth-two calculation for these representatives.

By translation, take the threshold interval to be
\[
M=\{0,1,2,3\}.
\]
After two Floquet periods, the forward causal cone of \(M\) is the
ten-site interval
\[
\{-3,-2,\ldots,6\}.
\]
Any output Pauli observable of weight at most two whose
backward image is supported in \(M\) must be one of
\[
3\cdot 10+3^2\binom{10}{2}=435
\]
nonidentity Pauli strings of weight one or two supported in this causal
cone.

For an output Pauli string \(P\), write
\[
\mathcal B_S^{(2)}(P)
=
\mathcal U_S^{-2}P\mathcal U_S^2
\]
for its backward image through two Floquet periods, modulo phase.
Applying the four inverse brickwork layers to the \(435\) candidates
gives exactly the following strings whose backward images are contained
in \(M\):
\[
\begin{array}{c|c|c}
\text{representative}
&
P
&
\mathcal B_S^{(2)}(P)
\\
\hline
S_{00}
&
Z_{-1}X_0
&
X_0Z_1I_2Z_3
\\
S_{00}
&
X_3Z_4
&
Z_0I_1Z_2X_3
\\[1mm]
S_{03}
&
Y_{-1}X_0
&
X_0Z_1Z_2X_3
\\
S_{03}
&
X_3Z_4
&
Z_0I_1Z_2Y_3
\\[1mm]
S_{33}
&
Y_{-1}X_0
&
Y_0Z_1Z_2X_3
\\
S_{33}
&
X_3Y_4
&
X_0Z_1Z_2Y_3.
\end{array}
\]
No other weight-one or weight-two candidate among the \(435\) strings
has backward support contained in \(M\).  Thus each representative has
exactly two threshold generators at depth two.  In the compressed
two-letter notation for their nonidentity output factors, these are
\[
\begin{array}{c|c}
S_{00}&XZ,\ ZX\\
S_{03}&XZ,\ YX\\
S_{33}&XY,\ YX.
\end{array}
\]

For each representative, the two displayed output generators have
disjoint support and therefore commute.  Clifford conjugation preserves
commutation, so their backward images also commute.  They are
independent modulo phases and hence generate an abelian Pauli algebra
isomorphic to
\(
\mathbb C^4.
\)
Their product has output weight four for the displayed representatives.

The two inward endpoint words of every backward image in the table
belong to the stable parity-tagged alphabet
\[
\mathcal B=\{ZI,XZ,YZ\}.
\]
By the boundary update in \Cref{app:boundary}, one further backward
Floquet period applies the permutations \(\pi_L\) and \(\pi_R\) to the
two endpoint words and moves each endpoint outward by two sites.
Therefore each of the two depth-two generators extends uniquely to a
threshold generator at every depth \(t\ge2\), and its input span grows
by four per period.

Conversely, equality in the minimum-span bound
\[
d_2(t)=4t-4
\]
forces equality at every expanding half-layer in the proof of
\Cref{lem:halflayer}.  Hence a threshold-saturating backward image must
remain in the stable boundary alphabet after its initial transient.
Since \(\pi_L\) and \(\pi_R\) are permutations of \(\mathcal B\), one
may reverse the boundary update and remove two exterior sites from
each end.  Repeating this reduction maps every threshold generator
injectively to a depth-two threshold generator.  Because the depth-two
calculation has exactly two such generators, no third generator can
appear at any later depth.

Finally, local Pauli-frame changes, translation, and reflection preserve
output weight, input span, independence, and commutation.  The result
therefore holds for all \(144\) bilaterally nilpotent classes.
\end{proof}

\section{Scalar witnesses for the qubit phase}
\label{sec:witness}

Set
\[
 x_\star=\frac1{\sqrt3}(1,1,1)^T,
 \qquad
 O_\star=\frac{X+Y+Z}{\sqrt3},
\]
and define
\[
 w_\nu=x_\star^TT_\nu^2x_\star,
 \qquad \nu\in\{L,R\}.
\]

\begin{theorem}
\label{thm:scalarwitness}
For every promised two-qubit Clifford dual-unitary edge channel,
\[
\begin{array}{c|c|c}
\Frob{T_\nu^2}^2&w_\nu&\text{edge type}\\
\hline
0&0&\text{nilpotent}\\
1&1/3&\text{rank-one glider}\\
3&-1/3\text{ or }1&\text{SWAP transport}.
\end{array}
\]
Therefore the pair $(w_L,w_R)$ determines the complete phase and the exact value
\[
 \Gamma\in\{0,1,2,6\}.
\]
The sign of $w_\nu$ must be retained.
\end{theorem}

\begin{proof}
For an iSWAP-core edge, $T=\pm e_ae_b^T$.  If $a\ne b$, then $T^2=0$.  If $a=b$, then $T^2=e_ae_a^T$ and $w_\nu=1/3$.

For a SWAP-core edge, $T$ is an orientation-preserving signed permutation, hence a rotation of the cube.  Its square either fixes the body diagonal $x_\star$ or maps it to one of the other tetrahedral body diagonals.  Their inner products with $x_\star$ are $1$ and $-1/3$.
\end{proof}

\begin{theorem}
\label{thm:maximin}
Among all scalar edge queries
\[
 W_{x,y}(T)=y^TT^2x
\]
with unit Bloch vectors $x,y$, the largest possible worst-case separation between the zero edge and every nonzero promised edge type is $1/3$.  The choice $x=y=x_\star$ attains this bound and simultaneously separates the SWAP alphabet.
\end{theorem}

\begin{proof}
The promised family contains the three coordinate projections $P_a=e_ae_a^T$.  Hence
\[
 \inf_{T^2\ne0}|W_{x,y}(T)|
 \le\min_a|x_ay_a|
 \le\frac13\sum_a|x_ay_a|
 \le\frac13.
\]
The uniform direction attains equality for every rank-one projection, and \Cref{thm:scalarwitness} shows that it also separates the SWAP sector.
\end{proof}

Operationally, one prepares the state with Bloch vector $x_\star$, applies the induced edge channel twice with the unused gate input freshly maximally mixed at each application, and measures $O_\star$.  One experiment is performed in each edge direction.  The tetrahedral state $x_\star$ and the observable $O_\star$ are not Pauli stabilizer resources.  The experimental tradeoff is therefore
\[
\begin{array}{c|c}
\text{protocol}&\text{required resources}\\
\hline
\text{two signed scalar settings}
&
\text{one tetrahedral preparation and measurement per edge}
\\
\text{Pauli-only certification}
&
\text{multiple Pauli preparation and measurement settings}
\end{array}
\]
Full Pauli transfer-matrix tomography uses nine scalar expectation
values per edge.  A smaller Pauli-only certification protocol would
require a separate construction and proof.

\begin{corollary}
\label{cor:shots}
With exact edge channels, $N$ independent $\{\pm1\}$ measurements per direction give total failure probability at most
\[
 4e^{-N/72}.
\]
Thus
\[
 \boxed{N\ge72\log\frac4\alpha}
\]
shots per direction suffice for failure probability at most $\alpha$.

More generally, suppose $\norm{\widetilde T_\nu-T_\nu}_2\le\varepsilon$ and both matrices are contractions.  If the statistical estimation error is at most $\eta$, nearest-alphabet decoding is guaranteed whenever
\[
 2\varepsilon+\eta<\frac16.
\]
A sufficient shot bound is
\[
 \boxed{
 N\ge\frac{2}{(1/6-2\varepsilon)^2}\log\frac4\alpha,
 \qquad \varepsilon<\frac1{12}.}
\]
\end{corollary}

\begin{proof}
The matrix perturbation obeys
\[
 \norm{\widetilde T_\nu^2-T_\nu^2}_2
 \le\norm{\widetilde T_\nu-T_\nu}_2
 (\norm{\widetilde T_\nu}_2+\norm{T_\nu}_2)
 \le2\varepsilon.
\]
The exact one-edge alphabet has minimum separation \(1/3\).
Applying Hoeffding's inequality to the independent measurement outcomes,
which take values in \([-1,1]\), and separately to their negatives gives
\cite[Theorem~2, Eq.~(2.6), p.~16]{hoeffding1963probability}
\[
\Pr\!\left(
  \left|\widehat w_\nu-w_\nu\right|\ge \eta
\right)
\le
2\exp\!\left(-\frac{N\eta^2}{2}\right).
\]
Taking \(\eta=1/6-2\varepsilon\) and applying a union bound over the two
edge directions proves the stated bounds.
\end{proof}

The witness theorem requires an independently justified Clifford dual-unitary promise.  Outside that discrete family, identical one-site edge data can coexist with different higher-body dynamics.

\section{Discussion and outlook}

Ordered local Clifford equivalence and homogeneous-circuit transport
are distinct.  Ordered local Clifford equivalence describes the nonlocal content of an isolated gate, whereas homogeneous brickwork transport also depends on how the local factors are repeated.  This is why the entire rank-one sector forms one ordered local orbit but produces delayed-erasure, one-glider, and two-glider circuits.  The one-site edge channels retain this wiring information and completely determine the one- and two-local masking laws.

The local dimension controls which transport phases are available.  A four-qubit perfect tensor does not exist \cite[Sec.~II.B and App.~C.2]{Rajchel2026}, so delayed erasure is optimal for qubits.  For \(q\ge3\), immediate bilateral erasure occurs and saturates the causal limits of the depth-\(2t\) brickwork architecture.  The counting formulas show that this perfect-tensor phase becomes dominant as \(q\) grows, while the delayed and gliding phases lie on the lower-dimensional determinantal and fixed-line incidence strata described in \Cref{thm:codimension}.  The same parameterization gives direct samplers and local-plus-core compilations for the finite-field Clifford gates.

The masking distances are exact local-privacy thresholds.  Below the relevant threshold, every one- or two-qudit observer receives a completely depolarizing channel, including for inputs entangled with a reference, while the complete output is decoded by reversing the circuit.  This is a low-locality access statement rather than a full threshold secret-sharing theorem because larger coalitions are not classified.  For qubits, the scalar edge witness gives a compact experimental test of the transport phase under the Clifford dual-unitary promise.

The main open problem is to determine the higher-local distances
\(d_r(t)\) for \(r\ge3\).  At weight three, edge-erasure depth is no
longer sufficient: qubit circuits separate according to boundary
monodromy, and canonical perfect cores already depend on the
cross-ratio parameter.  A complete description therefore requires
additional finite boundary invariants.  A second open problem is to
relate these higher-local profiles to quantum convolutional,
burst-error, or subsystem codes.

\appendix

\section{The qubit parity-tagged boundary automaton}
\label{app:boundary}
An exhaustive enumeration of the \(144\) bilaterally nilpotent
two-qubit symplectic classes partitions their parity-tagged boundary
dynamics into three types, represented by

\[
S_{00}=
\begin{pmatrix}
1&0&0&1\\
0&0&1&0\\
0&1&1&0\\
1&0&0&0
\end{pmatrix},
\quad
S_{03}=
\begin{pmatrix}
1&0&1&1\\
0&0&1&0\\
0&1&1&0\\
1&0&0&0
\end{pmatrix},
\quad
S_{33}=
\begin{pmatrix}
1&0&1&1\\
0&0&1&0\\
1&1&1&0\\
1&0&0&0
\end{pmatrix}.
\]
A brickwork circuit is invariant under translation by two sites, so endpoint parity is part of the boundary state.  For a finite Pauli string $P$, read the first two symbols inward from each end.  The stable alphabet is
\[
 \cB=\{ZI,XZ,YZ\},
 \qquad
 \tau:ZI\mapsto XZ\mapsto YZ\mapsto ZI.
\]
The stable left endpoint is even and the stable right endpoint is odd.

\begin{lemma}
Suppose the two inward endpoint words lie in $\cB$ with stable endpoint parities.  One backward Floquet period moves each endpoint outward by exactly two sites.  The boundary words update as
\[
\begin{array}{c|cc}
\text{representative}&\pi_L&\pi_R\\
\hline
S_{00}&\id&\id\\
S_{03}&\id&\tau\\
S_{33}&\tau&\tau.
\end{array}
\]
Hence every later period adds four to the span.
\end{lemma}

\begin{proof}
Only the two inverse gates meeting an endpoint can determine the two newly created outer symbols.  Direct multiplication by the three inverse binary symplectic matrices gives
\[
S_{00},S_{03}:\quad ZI\mapsto ZI,
\quad XZ\mapsto XZ,
\quad YZ\mapsto YZ
\]
at the left endpoint, whereas
\[
S_{33}:\quad ZI\mapsto XZ,
\quad XZ\mapsto YZ,
\quad YZ\mapsto ZI.
\]
The analogous right-end calculation gives the table above.  The update is independent of the interior Pauli labels.  Omitting endpoint parity is incorrect because the first backward step from the opposite sublattice can move an endpoint by only one site before it enters the stable class.
\end{proof}

For one- and two-local masking, the general half-layer proof in \Cref{lem:halflayer} replaces the earlier finite-base argument.  The parity-tagged automaton remains useful for the threshold leakage algebra and for the three-local spectrum.  At span four and depth two, exactly two weight-two directions attain the minimum, and they commute.

\section{Three-local qubit monodromy}
\label{app:d3}

Define the unordered boundary-monodromy invariant
\[
 \Omega(C)
 =
 \{\operatorname{ord}(\pi_L),\operatorname{ord}(\pi_R)\}.
\]
The \(144\) bilaterally nilpotent classes split as
\[
\begin{array}{c|c}
\Omega&\text{number of classes}\\
\hline
\{1,1\}&36\\
\{1,3\}&72\\
\{3,3\}&36.
\end{array}
\]

Exact binary symplectic enumeration gives
\[
\begin{array}{c|ccc}
\Omega&d_3(2)&d_3(3)&d_3(4)\\
\hline
\{1,1\}&2&5&8\\
\{1,3\}&2&6&10\\
\{3,3\}&3&6&10.
\end{array}
\]

It is enough to treat the representatives
\(S_{00},S_{03},S_{33}\), since local Pauli-frame changes,
translation by two sites, and reflection preserve the support
span and output weight. We enumerated all \(144\) bilaterally nilpotent symplectic classes.
For each class and each \(t\in\{2,3,4\}\), we enumerated all
nonidentity output Pauli strings of weight at most three modulo
two-site translation.  The three matrices
\(S_{00},S_{03},S_{33}\) represent the three resulting distance
profiles.

The minimizing histories found in this calculation enter the stable
boundary alphabet of \Cref{app:boundary}, after which the boundary
closure lemma increases their span by four per additional Floquet
period.  This supplies the corresponding upper bounds at all later
times.  A proof that no different transient weight-three history
produces a smaller span at an arbitrarily later time would require
a complete finite-state exhaustion of all transient boundary
histories.

\begin{conjecture}
\label{conj:d3qubit}
For every bilaterally nilpotent two-qubit Clifford dual-unitary gate,
\[
\begin{array}{c|c}
\Omega&d_3(t)\\
\hline
\{1,1\}
&
d_3(2)=2,\quad d_3(3)=5,\quad
d_3(t)=4t-8\quad(t\ge4)
\\[1mm]
\{1,3\}
&
d_3(t)=4t-6\quad(t\ge2)
\\[1mm]
\{3,3\}
&
d_3(2)=3,\quad
d_3(t)=4t-6\quad(t\ge3).
\end{array}
\]
\end{conjecture}

\section{Higher-local dependence on the perfect-tensor modulus}
\label{app:harmonic}

\subsection{A harmonic three-site cancellation}

The cross-ratio coordinate already appears in a concrete three-local cancellation.  This does not yet classify $d_3$ for all perfect gates, because arbitrary local decorations need not telescope in a homogeneous brickwork circuit.

\begin{proposition}
\label{prop:harmoniccollision}
Assume $\operatorname{char}\F_q\ne2$ and let $C_\delta$ be the canonical perfect core corresponding to $S_\delta$.  Place three $x$-type output Weyl labels at sites $1,2,4$ with amplitudes
\[
 (-2,2,1).
\]
After two backward Floquet periods, the labels outside \([4,7]\)
occur at sites \(0,1,2,3\) and are, in that order,
\[
 \frac{\delta-2}{\delta},
 \qquad
 \frac{(\delta-2)(\delta-1)}{\delta^2},
 \qquad
 \frac{(\delta-2)(3\delta-1)}{\delta^2},
 \qquad
 \frac{3(\delta-2)(\delta-1)}{\delta^2}.
\]
At \(\delta=2\), all four exterior labels vanish, while the remaining
nonzero labels lie at sites \(4,6,7\).  In particular,
\[
 d_3^{C_2}(2)\le4.
\]
Under \eqref{eq:crossratio}, \(\delta=2\) gives
\[
\lambda=\frac{2}{2-1}=2.
\]
For \(\operatorname{char}\F_q\ne2\), this is one of the three
anharmonic representatives
\[
\left\{-1,\,2,\,\frac12\right\}
\]
of a harmonic tetrad
\cite[Sec.~6.1]{Hirschfeld_projective}.
\end{proposition}

\begin{proof}
The $x$ coordinates evolve independently under the inverse-gate matrix
\[
 R_x=
 \begin{pmatrix}
 1&1\\[1mm]
 1&\dfrac{\delta-1}{\delta}
 \end{pmatrix}.
\]
Apply $R_x$ on the four alternating bonds encountered in two backward periods.  Substitution of the amplitudes $(-2,2,1)$ gives the displayed exterior labels.  At $\delta=2$, the remaining nonzero labels lie at sites $4,6,7$.
\end{proof}

\begin{conjecture}
\label{conj:d3harmonic}
For canonical perfect cores in odd characteristic,
\[
 d_3^{C_\delta}(t)=
 \begin{cases}
 2,&t=1,\\
 4t-4,&\delta=2\text{ and }t\ge2,\\
 4t-2,&\delta\ne2.
 \end{cases}
\]

The harmonic cancellation in
\Cref{prop:harmoniccollision} suggests this all-time law.  A proof
would require a complete symbolic classification of three-site
boundary histories.  For general locally decorated gates, additional
circuit-compatible boundary-monodromy data may be required beyond
\(\delta\).
\end{conjecture}

\begingroup
\small
\bibliographystyle{unsrtnat}
\bibliography{ref_cliff}
\endgroup

\end{document}